\documentclass[aps,prd,10pt,notitlepage,nofootinbib,superscriptaddress,showkeys,showpacs]{revtex4-1}

\usepackage{amstext,amsmath,amssymb,amsfonts,bbm}
\usepackage[latin1]{inputenc}
\usepackage{epsfig}
\usepackage{hyperref}
\usepackage{amsthm}
\usepackage{subfigure}
\usepackage{color}
\usepackage{multirow}

\usepackage{latexsym}

\usepackage{ulem}

\usepackage[all,arc,dvips,arrow,curve,cmactex]{xy}
\usepackage{amsfonts,amsmath,amssymb,graphics,psfrag}

\usepackage{cancel}

\renewcommand{\theequation}{
\arabic{equation}}

\def\be{\begin{equation}}
\def\ee{\end{equation}}
\def\ba{\begin{eqnarray}}
\def\ea{\end{eqnarray}}

\def\Sets{{\bf Sets}}
\def\op{{\rm op}}
\newcommand\ps[1]{\underline{#1}}

\def\Sub{{\rm Sub_{cl}}}
\def\us{\underline{\Sigma}}

\def\tr{{\rm tr }}
\def\Tr{{\rm Tr }}

\def\mb{\mathcal{B}}
\def\t{\mathbb{T}}
\def\tob{\ps{\mathbb{T}}}

\def\mn{\mathcal{N}}
\def\mh{\mathcal{H}}
\def\mv {\mathcal{V}}
\def\cV {\mathcal{V}}

\def\mP{\mathcal{P}}
\def\c{\mathcal{C}}

\def\cvn {\mathcal{V}(\mn)}

\def\cvnf {\mathcal{V}_{\rm f}(\mn)}
\def\cvfm {\mathcal{V}_{\rm f}(\mn)^-}
\def\d{\mathcal{D}}

\def\cvh {\mathcal{V}(\mh)}

\newcommand\Setn{{\Sets^{\cvn^{\op}}}}
\newcommand\Setnf{{\Sets^{\cvnf^{\op}}}}

\newcommand\dV{\downarrow\!V}
\newcommand\dase{\delta^{\text{o}}}

\newcommand\at{\alpha_{t}}

\newcommand\Hom{{\rm Hom}}
\newcommand\Autr{{\rm Aut}_{\rho}}
\newcommand\Aut{{\rm Aut}}

\newcommand\autf{{\rm Aut}_{\rm \,fth}}
\newcommand\Sh{{\rm Sh}}

\newcommand\Shvn{{\rm Sh}(\cvn)}
\newcommand\Shvfm{{\rm Sh}(\cvnf)}

\newtheorem{Theorem}{Theorem}[section]
\newtheorem{Definition}{Definition}[section]
\newtheorem{Lemma}{Lemma}[section]

\newtheorem{Proposition}{Proposition}[section]

\def\Rl{{\mathchoice
{\setbox0=\hbox{$\displaystyle\rm R$}\hbox{\hbox to0pt
{\kern0.4\wd0\vrule height0.99\ht0\hss}\box0}}
{\setbox0=\hbox{$\textstyle\rm R$}\hbox{\hbox to0pt
{\kern0.4\wd0\vrule height0.9\ht0\hss}\box0}}
{\setbox0=\hbox{$\scriptstyle\rm R$}\hbox{\hbox to0pt
{\kern0.4\wd0\vrule height0.9\ht0\hss}\box0}}
{\setbox0=\hbox{$\scriptscriptstyle\rm R$}\hbox{\hbox to0pt
{\kern0.4\wd0\vrule height0.9\ht0\hss}\box0}}
}}
\def\Cl{{\mathchoice
{\setbox0=\hbox{$\displaystyle\rm C$}\hbox{\hbox to0pt
{\kern0.4\wd0\vrule height0.9\ht0\hss}\box0}}
{\setbox0=\hbox{$\textstyle\rm C$}\hbox{\hbox to0pt
{\kern0.4\wd0\vrule height0.9\ht0\hss}\box0}}
{\setbox0=\hbox{$\scriptstyle\rm C$}\hbox{\hbox to0pt
{\kern0.4\wd0\vrule height0.9\ht0\hss}\box0}}
{\setbox0=\hbox{$\scriptscriptstyle\rm C$}\hbox{\hbox to0pt
{\kern0.4\wd0\vrule height0.9\ht0\hss}\box0}}}}

\begin{document}

\title{Topos Analogues of the KMS State}

\author{{\bf Joseph Ben Geloun}}
\email{jbengeloun@perimeterinstitute.ca}
\affiliation{Perimeter Institute for Theoretical Physics, 31 Caroline
St, Waterloo, ON, Canada}
\affiliation{International Chair in Mathematical Physics and Applications, 
ICMPA-UNESCO Chair, 072BP50, Cotonou, Rep. of Benin}
 
\author{{\bf Cecilia Flori}}
\email{cflori@perimeterinsitute.ca}
\affiliation{Perimeter Institute for Theoretical Physics, 31 Caroline
St, Waterloo, ON, Canada}

\date{\small\today}

\begin{abstract}
\noindent  
We identify the analogues of KMS state in topos theory. 
Topos KMS states can be viewed as {\it classes} of truth objects
associated with a measure $\mu^\rho$ (in one-to-one correspondence 
with an original KMS state $\rho$) which satisfies topos condition analogues to the ordinary KMS conditions. 
Topos KMS states can be defined both {\it externally} or {\it internally} according to whether the automorphism group
of geometric morphisms acts {\it externally} on the topos structure, or the automorphism group is an object in the topos itself, in which case, it acts {\it internally}. 
We illustrate the formalism by a simple example on the Hilbert 
space $\Cl^3$.

\medskip

\noindent Pacs: 02.10.-v, 02.10.Ab, 02.10.De
\noindent pi-mathphys-286 and ICMPA-MPA/2012/11

\end{abstract}

\maketitle

\setcounter{footnote}{0}

\tableofcontents

\section{Introduction}
\label{intro}

Topos formulation of quantum theory was put forward by 
Isham,  Butterfield, D\"oring and co-workers
\cite{isham1,isham2,isham3, isham4,isham5,isham6,andreas,andreas1,andreas2,andreas3,andreas4,
andreas5,andreas6,jordan,unsharp,quantization,consistent2,consistent3,mixed, dasain,cecilia,cecilia2,cecilia3,foundations} as an attempt to solve certain interpretational problems inborn in quantum theory and mainly due to the mathematical representation of the theory \cite{ceciliabook}.
In particular, the main idea stressed by the authors, in the above-mentioned
papers, is that the use of topos theory to redefine the mathematical
structure of quantum theory leads to a reformulation of quantum theory in such a way that it is made to `look like' classical
physics. Furthermore, this reformulation of quantum theory has the
key advantages that (i)  no fundamental role is played by the
continuum; (ii) propositions can be given truth values without
needing to invoke the concepts of `measurement' or `observer'; 
and (iii) probabilities are given by a logical definition, that is to say the notion of probabilities becomes secondary while the notion of truth values becomes primary.

The reasons why such a reformulation is needed in the
first place concern quantum theory in general and quantum
cosmology in particular.
\begin{itemize}
\item As it stands, quantum theory is {\it non-realist}.\footnote{By a {\it realist} theory, we mean a theory in which the following conditions are satisfied: (i) propositions form a Boolean algebra; and (ii)
propositions can always be assessed to be either true or false. As
will be delineated in the following, in the topos approach to
quantum theory both of these conditions are relaxed, leading to
what Isham and D\"oring called  a {\it neo-realist} theory \cite{andreas6}.} From a mathematical perspective, this is reflected in the Kocken-Specher theorem.\footnote{
\textbf{Kochen-Specker Theorem}: If the dimension of the Hilbert space $\mh$ is
greater than 2, then there does not exist any valuation function
$V_{\vec{\Psi}}:\mathcal{O}\rightarrow\Rl$  (depending on $\vec\Psi$ a vector in $\mh$) from the set
$\mathcal{O}$ of all bounded self-adjoint operators $A$ of
$\mh$ to the reals $\Rl$ such that,  for all
$A\in\mathcal{O}$ and all $f:\Rl\rightarrow\Rl$, the following holds $V_{\vec{\Psi}}(f(\hat{A}))=f(V_{\vec{\Psi}}(\hat{A})).$} 
\item Notions of {\it measurement} and {\it external observer} introduce 
issues  when dealing with cosmology. 
\item The existence of the Planck scale suggests that there is {\it a priori} no  justification for the adoption of the notion of a continuum in the quantum theory used in formulating quantum gravity.
\end{itemize}

These considerations led Isham and D\"oring to search for a
reformulation of a more realist quantum theory. It turns out that this can be achieved through
the adoption of topos theory, as the mathematical framework with
wish to reformulate  the notion of quantum theory.
The strategy adopted to attain such a new formulation is to
re-express quantum theory as a type of classical theory in a
particular topos.
The notion of classicality in this setting is defined in terms of the notion of {\it context} or {\it classical snapshots}. In this framework, quantum theory is seen as a collection of {\it local} classical snapshots, where the quantum information is determined by the relation of these local classical snapshots.

Mathematically, each classical snapshot is represented by an
abelian von-Neumann sub-algebra $V$ of the algebra
$\mathcal{B}(\mh)$ of bounded operators on a Hilbert space. The
collection of all these contexts forms a category $\cvh$
which is actually a poset by inclusion.
As one goes to smaller sub-algebras $V'\subseteq V$, one obtains a coarse-grained classical perspective on the theory.
The fact that the collection of all such classical snapshots forms a category, in particular a poset, means that the quantum information can be retrieved by the relations of such snapshots, i.e. by the categorical structure.

A topos that allows for such a classical local description is the
topos of presheaves over the category $\mv(\mh)$. This is denoted
as $\Setn$. By using the topos $\Setn$ to
reformulate quantum theory, it was possible to define pure quantum
states, quantum propositions and truth values of the latter
without any reference to external observer, measurement or any
other notion implied by the instrumentalist interpretation. In
particular, for pure quantum states, probabilities are replaced by
truth values which derive from the internal structure of the
topos itself.
These truth values are lower sets in the poset $\mv(\mh)$, thus they are interpreted as the collection of all classical snapshots for which the proposition is true. Of course, being true in one context implies that it will be true in any coarse graining of the latter.

As just mentioned, the mathematics underpinning the definition of 
classical  sub-components of a quantum theory is  a
topos  over the category of abelian von Neumann sub-algebras
of bounded operators over a given Hilbert space. 
The von Neumann algebras are the preferred fields 
for defining another important object in quantum mechanics 
and statistical physics: the KMS (Kubo-Martin-Schwinger) state \cite{Kubo:1957mj}\cite{Martin:1959jp}. 

The KMS state occurs in statistical analysis
 as a state describing the thermal equilibrium
of many body quantum systems \cite{Kubo:1957mj}\cite{Martin:1959jp}. 
For oscillator systems and under particular conditions, the  KMS state turns out to coincide with Gibbs state \cite{ion,taku,sewell}. Its significance has also been highlighted through a mathematical framework called the Tomita-Takesaki modular theory on von Neumann algebras \cite{stratila,takesaki, takesaki2, maeda, Borchers:2000pv, Summers:2003tf,Connes:1994yd}. The Tomita-Takesaki theory on its own was the focus on intense research activities in the mid 70's and after. 
Indeed, it has been revealed that, for algebras of observables associated with a quantum field theory obeying Wightman axioms,  the modular objects associated with the vacuum state and algebras of observables, localized in certain wedge-shaped regions in Minkowski space, have the
some geometric content (the unitary modular group
is associated to Lorentz boost symmetry, whereas the modular involution 
implements space-time reflection about the edge of the wedge and
charge conjugation) \cite{Borchers:2000pv,Summers:2003tf}.
Thus, KMS state and Tomita-Takesaki theory are significant in physics.

 In this paper, we succeed in identifying the topos analogues
of the KMS state or KMS condition for an equilibrium state. 
In fact, given a (modular) automorphism in the Hilbert space of states,
to which there corresponds a geometric morphism in the topos formulation,
 there can be defined two  KMS conditions:
i) the {\it external } condition, defined through 
a geometric morphism, which does not respect the truth value
object (topos state) fibration but maps the topos KMS state
onto a twisted KMS presheaf. ii) The {\it internal} condition,
defined via an internal group object and which respects the 
presheaf structure of the topos state. 
Roughly speaking, a topos KMS state can be viewed 
as a truth object presheaves for which, the KMS condition is satisfied on the sub-objects at each context. More precisely, the topos KMS state will be  defined as a collection of sub-objects of the state space such that their measure, as defined for each context, is invariant with respect to a given automorphism of the von-Neumann algebra.

When defining the measures on the spectral presheaf and
seeking a characterization of  measures
corresponding to normal states, D\"oring 
mentioned a likely link between KMS states and topos theories \cite{doring}.
The present study, never addressed before
to the best of our knowledge, provides
explicitly this link and, thereby,  illustrates once
again the wealth of Topos quantum theory. 

 This paper also provides an interesting illustration of the possible uses of the internal group transformations in topos quantum theory defined in \cite{group,gelfandII}. There group transformations were defined in such a way so as to solve the problem of twisted presheaves (see \cite{andreas5}) thus giving a well-defined meaning to the concept of group action in topos quantum theory. The results of this paper show that the definition of group transformations of \cite{group} also provides a solution for the problem of defining the topos analogue of KMS state. In fact, these group transformations are given in terms of automorphisms on the topos analogue of the state space and one of the KMS conditions is the state's invariance under an automorphisms of the algebra to which it belongs. However, since topos states are particular collection of sub-objects of the state space, such a KMS condition can be reproduced at the level of topos quantum theory only if group transformations are defined as automorphisms on the state space.

It should be noted, at this point, that by defining the KMS condition via the above mentioned automorphism on the state space we are able to define the KMS condition globally in the sense that we are defining it for all abelian sub-algebras of a non commutative von Neumann algebra $\mn$. As explained later in the paper, such a collection of abelian von Neumann sub-algebras form the base category $\mv(\mn)$ of the topos utilised to express quantum theory. 
In this sense, the definition of KMS state present in this paper is different from the definition of classical KMS condition put forward in \cite{gallavotti1, gallavotti2} where the authors only consider abelian algebras of observables. In particular in \cite{gallavotti1} the authors define a classical KMS boundary condition while in \cite{gallavotti2}, they define a non trivial time evolution in terms of the KMS condition so that, also in the classical case, a non trivial instance of the Tomita-Takesaki theorem can be defined. 
Given that these works only deal with commutative algebras, it would be very interesting to see how and if it is possible to translate them in the covariant topos approach to quantum theory \cite{bass, bass1} rather than the contravariant approach dealt with in the present paper. The reasoning being that in the covariant approach the starting point is an internal (to the topos) abelian $C^*$-algebra which is precisely the starting point in \cite{gallavotti1, gallavotti2}. However this line of ideas is beyond the scope of this paper and will deserve full investigation.

The plan of the paper is the following: The next 
section briefly yields basic definition on von Neumann 
algebra, the simplest instance referring to the KMS
condition on a state and a short summary of the modular theory 
of Tomita-Takesaki \cite{takesaki}. In Section \ref{sect:toposkms}, we start
our analysis and define the topos analogue of the 
KMS state under an {\it external} condition. 
We then identify its properties. 
Section \ref{sect:inter} is devoted to the construction 
of the topos KMS state subjected now to an {\it internal}
condition. 
 Section \ref{sect:tttheo}
deals, in a streamlined analysis, with the extension of the Tomita-Takesaki modular structure  in topos formulation.  
Finally, a detailed appendix gathers basic definitions 
on category and topos theories and important statements used throughout the text.

\section{KMS state and the Tomita-Takesaki modular theory}
\label{sect:tt-theo}

We briefly provide here some definitions related to 
von Neumann algebras, the canonical KMS  state
and  the Tomita-Takesaki modular theory (see, for example,
\cite{sewell,stratila,takesaki}).

\medskip 
\noindent{\bf On von Neumann algebras.} Let $\mh$ be a complex separable Hilbert space of any dimension,
we denote by $\mb(\mh)$ the set of all bounded
operators on $\mh$. Let $\mathcal{A}$ be a subset of
$\mb(\mh)$ and $\mathcal{A}'$
its {\it commutant}, i.e. the set of all elements of  $\mb(\mh)$ commuting with every element of $\mathcal{A}$.

Assuming that (1) $\mathcal{A}$ be closed under linear combinations,
(operator) multiplication and adjoint conjugation,
and (2) $\mathcal{A} = \mathcal{A}''$,
then $\mathcal{A}$ is called a {\it von Neumann algebra}.
A von Neumann algebra always contains the identity operator $I_\mathcal{H}$ on $\mathcal{H}$.
This statement relies on the fact that $\mathcal{A}$ is a weakly closed set
with respect to some operator scalar product.

The von Neumann algebra $\mathcal{A}$ is called a {\it factor} if
the only common elements of $\mathcal{A}$ and its commutant
are those proportional to the identity of $\mathcal{H}$, i.e.
if $\mathcal{A} \cap \mathcal{A}' = \mathbb C I_\mathcal{H}$.

Given $\varphi : \mathcal{A} \longrightarrow
\mathbb C$, be a bounded linear functional on $\mathcal{A}$, its action
is given by duality scalar product denoted by
$\langle\varphi\,;\, A\rangle$ or equivalently by $\varphi(A)$, for $A\in \mathcal{A}$. The functional
$\varphi$ is called a {\it state} on this algebra if (a)
$\langle \varphi\,;\, A^*A\rangle \geq 0$, $\forall A \in \mathcal{A}$
and  (b) $\langle \varphi \, ; \, I_\mathcal{H}\rangle =1$.

\noindent$\bullet$ A state $\varphi$ is said to be {\it faithful}
if $\langle \varphi\,;\,A^*A\rangle > 0$ for all $A \neq 0$.

\noindent$\bullet$ A state is said to be {\it normal} if and only if there is a density matrix
$\varrho$, such that $\langle \varphi\mid A\rangle = \text{Tr}[\varrho A]$, $\forall
A\in \mathcal{A}$.

\noindent$\bullet$ A state is called a {\it vector state} if there exists a vector $\phi \in \mathcal{H}$,
such that $\langle \varphi\,;\,A\rangle = \langle\phi \mid A\phi\rangle$, $\forall
A \in \mathcal{A}$. Note that such a state is also normal.

\noindent$\bullet$ A vector $\psi \in \mathcal{H}$ is called {\it cyclic} for the von Neumann
algebra if the set $\{A\psi \mid A \in \mathcal{A}\}$ is dense in $\mathcal{H}$.

\noindent$\bullet$ A vector $\psi \in \mathcal{H}$ is said to be {\it separating } for $\mathcal{A}$
if $A\psi = B\psi$, $\forall A, B \in \mathcal{A}$, if and only if $A = B$.

\medskip 
\noindent{\bf Canonical KMS condition.}
The general notion of  KMS state in  a von-Neumann algebra $\mn$
 can be stated as follows:
\begin{Definition}[\cite{kms}]
\label{def:0kms}
Given a von Neumann algebra $\mathcal{N}$ equipped with a one parameter automorphisms group $\{\alpha(t)|\,t\in\Rl\}$, then a  normal state $\rho:\mathcal{N}\rightarrow \Cl$ is a KMS state if it satisfies the following properties:
\begin{enumerate}
\item Invariance under the automorphism group: $[\alpha(t)\cdot \rho](A)= \rho(\alpha(t)A)=\rho(A)$;
\item Given any two elements $A, B\in \mathcal{N}$, the function
\be
F_{A,B}(t)=\langle\rho;A\,\alpha(t)B\rangle\,,
\ee
for all $t\in \Rl$, has an extension to the strip $\{z=t+iy|\,t\in\Rl, y\in[0,\beta]\}\subset \Cl$, such that $F_{A,B}(z)$ is analytic in the open strip $(0,\beta)$ and continuous on its boundary. Moreover, it satisfies the boundary condition
\be
F_{A,B}(t+i \beta) = \langle\rho;\alpha(t)BA\rangle\,,
\label{fabit}
\ee
for $t\in \Rl$ and $A, B$ in $\mn$.
\end{enumerate}
\end{Definition}

\medskip

 As shown in Proposition 5.3.7 of \cite{kms} the above definition is equivalent to the following
\begin{Definition}\label{def:1kms}
Given a von Neumann algebra $\mn$ with an automorphism group $\{\alpha(t)|\,t\in\Rl\}$, a state $\rho$ on $\mn$ is a KMS state if
\be\label{equ:second}
\rho(A\alpha_{i\beta}(B))=\rho(BA)\,,
\ee
for all $A, B$ in $\mn_{\alpha}$, a $\sigma$-weakly dense $\alpha$-invariant $*$-sub-algebra of $\mn$.
\end{Definition}
From the equivalence of these two definitions it follows that for all $A, B\in\mn_{\alpha}$ then $F_{A,B}(t+i\beta)=\rho(A\alpha_{t+i\beta}B)=\rho(\alpha(t)BA)$.  This is precisely condition \eqref{equ:second} in fact $\rho(A\alpha_{t+i\beta}B)=\rho(A\alpha_{i\beta}(\alpha_t(B)))$, renaming $B'=\alpha_t(B)$, we obtain $\rho(A\alpha_{i\beta}(B'))=\rho(B'A)$.

\medskip

\noindent{\bf Tomita-Takesaki modular theory.}
Let $\mathcal{A}$ be a von Neumann algebra on a Hilbert space $\mh$ and $\mathcal{A}'$ its commutant.
Let $\Phi \in \mh$ be a unit vector which is cyclic and separating for $\mathcal{A}$.  There exists 
a corresponding state $\varphi=\varphi_{\Phi}$ on the algebra,
$\langle \varphi\, ; \, A\rangle =
\langle \Phi \mid A \Phi \rangle$, $A \in \mathcal{A}$, which
 is faithful and normal.
The map
\ba
  S:\mh\longmapsto \mh\; ,
\qquad SA\Phi = A^*\Phi\; , \; \forall A \in \mathcal{A} \; .
\label{mod-antilin-map}
\ea
is antilinear. Due to the fact that $\Phi$ is cyclic, $S$ is densely defined and closable.
Denoting its closure again by $S$, the latter admits a polar decomposition as
\ba
  S = J\Delta^\frac 12 = \Delta^{-\frac 12}J\; , \quad \text{with} \quad \Delta = S^*S\; .
\label{mod-pol-decomp}
\ea
The operator $\Delta$ is called the {\it modular operator}. It is positive and self-adjoint.
Moreover, $J$, called the {\it modular conjugation operator}, is antiunitary and obeys
$J = J^*$, $J^2 = I_\mathcal{H}$. By the antiunitarity of $J$, we have
$\langle J\phi | J\psi \rangle = \langle \psi | \phi\rangle$, $\forall \phi, \psi
\in \mathcal{H}$.

Using the fact that $\Delta$ is self-adjoint and therefore admits a spectral representation,
the family of operators $\Delta^{-i\frac {t}{ \beta}}$, for $t\in \mathbb R$
and for some fixed $\beta > 0$, is well defined. From this point, there exists
a unitary family $\{\alpha_\varphi (t)\}_{t\in \Rl}$ of automorphisms of the algebra $\mathcal{A}$, such that
\ba
  \alpha_\varphi (t)[A] = \Delta^{i\frac {t}{ \beta}}A\Delta^{-i\frac{t}{\beta}}\; , \;\;
  \forall A \in \mathcal{A}\; .
\label{mod-automorph-grp}
\ea
These automorphisms determine a strongly continuous one-parameter group of automorphisms,
called the {\it modular automorphism group\/}.
The generator $\mathbf H_\varphi$ of this one-parameter group is such that
$\Delta^{-i\frac {t}{\beta}} = e^{it\mathbf H_\varphi}$ or
$\Delta = e^{-\beta \mathbf H_\varphi}$.
It can be shown that the state $\varphi$ is invariant under this automorphism group:
$  e^{-\beta \mathbf H_\varphi}\Phi = \Phi$ (and that the von Neumann algebra is stable
$\Delta^{i\frac {t}{\beta}}\;\mathcal{A} \Delta^{-i\frac {t}{\beta}} = \mathcal{A}$).
Remarkably, the antilinear map $J$ interchanges $\mathcal{A}$ with its commutant $\mathcal{A}'$:
$J\mathcal{A} J = \mathcal{A}'$ and $J \mathcal{A}' J = \mathcal{A}$.

In this specific instance, 
using the modular evolution operator $\alpha_\varphi(t)$, 
the KMS condition can be easily translated in terms of the
modular structure $\Delta$ itself related to the vector $\Phi$,
 the state $\varphi$ or its associated density.

\section{Topos analogue of the KMS state}
\label{sect:toposkms}

In this section, we define the topos analogue of a KMS state. The existence of such a state is granted by the 1:1 correspondence between states $\rho$ and truth objects $\ps{\mathbb{T}}^{\rho}$, which represent the topos analogue of a state. Since the definition of $\ps{\mathbb{T}}^{\rho}$ depends on the definition of a measure $\mu^{\rho}$ on the state space $\us$, we  first need to introduce such a measure. 

\subsection{Measure on the state space}
\label{subsect:mes}

Let $\mn$ be a von Neumann algebra and $\cvn$ be the category\footnote{A short survey on category theory is provided in 
Appendix \ref{app:category}.}  of abelian von Neumann sub-algebras of $\mn$ ordered by sub-algebra inclusion. The elements $V\in\cvn$ will be called {\it contexts}.

For now, let us consider $\Sets$ the category of sets ordered by subset inclusion
and the topos of contravariant functors from $\cV(\mn)$ to $\Sets$ which we denote as $\Setn$ 
(see Appendix \ref{app:stat}, Definition \ref{appdefsetn}).

To start the analysis, we need the notion of measure and of 
measurable subsets of the state space in the topos $\Setn$ \cite{andreas}. 
The state space $\us$ is defined as an object of $\Setn$
such that, to each $V$ it assigns the set $\us_V=\{\lambda:V\to \Cl|\,\lambda(\hat{1})=1\}$ which is the Gel'fand
spectrum of $V$,
(see Appendix \ref{app:stat}, Definition \ref{appdefus} for 
more precisions). To have a proper notion of measure on $\us$, we need to define (I) the collection of measurable sets and (II) the object in which these measures take their values. These notions are defined respectively as: (I) the collection of clopen sub-objects of $\us$ which is denoted as  $\Sub(\us)$ and such that  for each $V\in \cvn$, $\Sub(\us)_V$ represents the lattice of clopen subsets of  $\us_V$. The lattice operations are given by intersection and union while the lattice ordering is given by subset inclusion (see Appendix \ref{app:stat}, Definition \ref{appdefsub}); (II) The
collection of global elements (global sections) $\Gamma\ps{[0,1]}$ 
which is a collection of natural transformations from the terminal object $\ps{1}\in\Setn$ to the presheaf $\ps{[0,1]}$ 
(see Appendix \ref{app:mes}, Definition \ref{appdefgam}). The presheaf $\ps{[0,1]} \in \Setn$ is the object in which the measure takes its value. In fact, it assigns to each 
context $V$ the set $\ps{[0,1]}_V$ of order reversing (which we denote by OR)
functions  from $\dV$ to $[0,1]$ 
(see Appendix \ref{app:mes}, Definition \ref{appdef01}).

In order to be able to define a measure, we need an ulterior feature, namely the existence, for each context $V$, of a 1:1 correspondence between the lattice $\mathcal{P}(V)$ of projection operators in $V$ and clopen subsets of $\us_V$. This correspondence is defined 
via the lattice homeomorphism 
\begin{equation}
\label{equ:smap}
\mathfrak{S}:\mathcal{P}(V)\rightarrow \Sub(\us)_V\hspace{.2in}
\end{equation}
such that
\begin{equation}
\hat{P}\mapsto
\mathfrak{S}(\hat{P}):=\ps{S}_{\hat{P}_V}:=\{\lambda\in\us_V|\lambda
(\hat{P})=1\}\,.
\end{equation}
Thus $\hat{P}_{\ps{S}_V}=\mathfrak{S}^{-1}(\ps{S}_V)$ (see Appendix \ref{subsect:proj}). 

We are in a position to define a measure on $\us$.  
\begin{Definition}[\cite{andreas}]
Given any state $\rho$ of $\mn$, the associated measure on $\Sub(\us)$ is
\ba
\label{ali:measure2}
\mu^{\rho}: \Sub(\us)&\rightarrow& \Gamma\ps{[0,1]}\\
\ps{S}=(S_V)_{V\in\cV(\mh)}&\mapsto&\mu^{\rho}(\ps{S}):=(\rho(\hat{P}_{\ps{S}_V}))_{V\in\cV(\mh)}\,.
\ea
\end{Definition}
The association between a state $\rho$ and a measure is 1:1 
(Appendix \ref{app:mes} provides more details on this correspondence). 
Let us now analyze the measure in \eqref{ali:measure2} in more detail. In particular, for any $V\in \cvn$, we obtain
\ba
[\mu^{\rho}(\ps{S})](V):\ps{1}_V&\rightarrow &\ps{[0,1]}_V\\
\{*\}&\mapsto&[\mu^{\rho}(\ps{S})](V)(\{*\}):=
[\mu^{\rho}(\ps{S})](V)
\ea
where $[\mu^{\rho}(\ps{S})](V):\,\dV\rightarrow [0,1]$.
Thus, the measure $\mu^{\rho}$ defined in \eqref{ali:measure2} takes a clopen sub-object of $\us$ and defines an OR function 
\be
\mu^{\rho}(\ps{S}):\cvn\rightarrow [0,1]
\ee
such that, for each $V\in \cV(\mn)$, $[\mu^{\rho}(\ps{S})](V)$ defines the expectation value with respect to $\varrho$ of the projection operator $\hat{P}_{\ps{S}_V}$ (associated with $\ps{S}_V$). Therefore, given two contexts $V'\subseteq V$, since $\hat{P}_{\ps{S}_{V'}}\geq \hat{P}_{\ps{S}_V}$ then $\rho(\hat{P}_{\ps{S}_{V'}})\geq \rho( \hat{P}_{\ps{S}_V})$. 
Results from \cite{gelfandI} and in particular Proposition 20 of \cite{gelfandII} show that $\mu^{\rho}_V=\rho_V\circ \mathfrak{S}^{-1}_V$ for all $V\in\mv(\mh)$. Moreover Lemma 23 in  \cite{gelfandII} shows that, given an automorphism $\alpha:\mn\rightarrow\mn$ it is possible to define a group action on the measure as follows
\be
\alpha \cdot \mu^{\rho}=\mu^{\alpha\cdot \rho}=\mu^{\rho\circ \alpha^{-1}}
\ee
such that $\alpha\cdot \mu^{\rho}$ is the measure uniquely associated to the state $\alpha\cdot \rho:=\rho\circ \alpha^{-1}$. These results will play a key role in KMS condition on the measure. \\

From the above discussion, given a KMS state $\rho:\mn\rightarrow\Cl$ we can deduce that there is a measure $\mu^{\rho}$ uniquely associated with it. In the following section, we will define some characterising properties of such a measure and, consequently, define the topos analogue of the KMS state.

\subsection{External KMS condition}
\label{subsect:kms}

From this section, our main results are discussed. 
 We are now interested in understanding the properties of the measure $\mu^{\rho}$ associated to a KMS state. To achieve this goal, first of all, we  analyze the automorphisms group $\Aut=\{\alpha(t) |\,t\in \Rl\}$ acting on $\mn$. 
Since each $\alpha(t)$ can be extended to a functor $\alpha(t):\cvn\rightarrow\cvn$; $V\mapsto\alpha(t)V$, this induces a geometric morphism\footnote{
A {\it geometric morphism} \cite{topos7}\cite{sv}  
$\phi$ between topoi $\tau_1$ and $ \tau_2$ is defined as a pair 
$(\phi_*, \phi^*)$ of functors such that 
$\phi_*:\tau_1\rightarrow \tau_2$ and
$\phi^*:\tau_2\rightarrow \tau_1$ called, respectively, the
{\it direct image} and the {\it inverse image} part of $\phi$. Moreover, 
(a) $\phi^*\dashv  \phi_*$ i.e., $\phi^*$ is the left adjoint of $\phi_*$;
(b) $\phi^*$ is left exact, i.e., it preserves all finite limits.
}
whose inverse image part is 
\ba
\alpha(t)^*:\Setn&\rightarrow&\Setn\crcr
\ps{A}&\mapsto&\alpha(t)^*\ps{A}
\ea
where, for each $V\in\cvn$,  $\alpha(t)^*\ps{A}_V:=\ps{A}_{\alpha(t)(V)}$.

Importantly, note that we have defined 
the automorphism group in an {\it external} way, 
i.e. $\Aut$ is not itself an object in the topos. Hence the ensuing
analysis and conditions that we investigate can be qualified
in this respect as being {\it external}. In the next section,
another type of automorphism group will be introduced.

The following statement holds:

\begin{Proposition}[KMS measure]
\label{pro:kms}
Given a KMS state $\rho$, for all elements of the one parameter group
$\rm{Aut}=\{\alpha(t)|\,t\in \Rl\}$, the associated measure $\mu^{\rho}$ satisfies the following conditions:
 \begin{enumerate}
\item [(C1)]$ \mu^{\rho}(\ps{S})=\mu^{\rho}(\alpha(t)^*\ps{S})$,
$\forall \ps{S} \in \Sub(\us)$;
\item [(C2)] $\forall \ps{S}, \ps{T} \in \Sub(\us)$ and $V\in \mv(\mh)$ the function
\be
F_{\ps{T},\ps{S}}(t):= [\mu^{\rho}(\ps{T}\wedge \alpha(z)^* \ps{S})](V)\,,
\ee
for all $t\in \Rl$ has an extension to the strip $\{z=t+iy|\, t\in\Rl, \,y\in[0,\beta]\}\subset \Cl$ such that $F_{\ps{T},\ps{S}}(z)$ is analytic in the open strip $(0,\beta)$ and continuous on its boundary where it satisfies the following boundary condition
 \be 
F_{\ps{T},\ps{S}}(t+i\beta)=\mu^{\rho}(\alpha^*(t)\ps{S}\wedge \ps{T} )(V)\,.
\label{equ:cong2}
\ee
\end{enumerate}

\end{Proposition}
From the definition of KMS states given in Definition \ref{def:1kms}, it follows that it is possible to give a precise characterisation of the function $F_{\ps{T},\ps{S}}(t+i\beta)$ for an appropriate domain. In particular, given any algebra $V$, we consider the $\sigma$-weakly dense $\alpha$ invariant *-sub-algebra $V_{\alpha}\subseteq V$. For these contexts, $F_{\ps{T},\ps{S}}(t+i\beta)$ can be defined as  
\be
F_{\ps{T},\ps{S}}(t+i\beta):=\mu^{\rho}(\ps{T}\wedge \alpha(t+i\beta)^* \ps{S})(V_{\alpha})\,,
\ee such that the boundary condition becomes
 \be
\mu^{\rho}(\ps{T}\wedge \alpha(t+i\beta)^* \ps{S})(V_{\alpha})=\mu^{\rho}(\alpha^*(t)\ps{S}\wedge \ps{T} )(V_{\alpha})\,.
\ee

Let us analyse conditions C1 and C2, separately.
For simplicity, in the following, we will write $\alpha(t)=\at$. 

 In order to study C1 and C2, we first  need to define how $\at$ acts on a state $\rho$. If $\rho$ is a normal state, then  it can be defined as $\tr(\varrho -)$, where $\varrho$ is a density matrix. We then obtain
\be
\at \cdot\rho=\at \cdot \tr(\varrho-):=\tr(\at\varrho -)=\tr(\varrho\alpha_t^{-1}-)=\rho\circ \alpha_t^{-1}
\ee
Thus, for a general state we can define
\be
\at \cdot \rho:=\rho\circ \at^{-1}\,.
\ee
This being introduced, we can turn to the question of the significance of  C1. The condition C1 means that, for all $V\in \cvn$, one has
\ba
[\mu^{\rho}(\at^*\ps{S})](V)&=&\mu^{\rho}[(\ps{S}\circ \at)|_V]:=\rho_{|V}\circ \mathfrak{S}_{|V}^{-1}(\at^*\ps{S})(V)\\
&=&\rho_{|V} (\hat{P}_{(\at^*\ps{S})(V)})=\rho_{|V}\circ \at\hat{P}_{\ps{S}_V}\\
&=& \at^{-1}\cdot \rho_{|V}\circ \mathfrak{S}_{|V}^{-1}\ps{S}_V\\
&=& [\mu^{\rho}(\ps{S})](V)=\rho_{|V}\circ \mathfrak{S}_{|V}^{-1}\ps{S}_V
\ea

Our claim is that condition C1 on the measure represents the first condition of the KMS state $\rho$ (see point 1. in Definition \ref{def:0kms}), namely that $\rho$ is invariant under the automorphism group. 
To proceed with this, let us prove that the following diagram commutes:
 \[\xymatrix{
\mP(V)\ar[rr]^{\mathfrak{S}_V}\ar[dd]_{\alpha(t)}&&\Sub(\us_V)\ar[dd]^{\alpha(t)}\\
&&\\
\mP(\alpha(t)V)\ar[rr]^{\mathfrak{S}_{\alpha(t)V}}&&
\Sub((\alpha^*(t)\us)_V)\\
}\]

The diagram in one direction yields, for a given $\hat{P}\in \mP(V)$, 
 \be
[\at\circ \mathfrak{S}_V](\hat{P})=\at(\ps{S}_{\hat{P}_V})=\{\at(\lambda)
\in\us_{\at V}|\; [\at\lambda](\hat{P})=1\}=
\{\lambda\in\us_{\at V}|\; \lambda(\at(\hat{P}))=1\}\,.
\ee

On the other direction, we get:
\be
[\mathfrak{S}_{\at V}\circ \at] (\hat{P})=\mathfrak{S}_{\at V}( \at(\hat{P}))=\ps{S}_{(\at\hat{P})_{\at V}}=\{\lambda\in\us_{\at V}|\;\lambda(\at(\hat{P}))=1\}\,.
\ee

Thus the diagram commutes. Using this, we can now prove 
conditions C1 and C2 in Proposition \ref{pro:kms}:
\begin{proof}[Proof of C1.]
Given a KMS state $\rho$, for any $V\in \cvn$, we have:
\ba
[\mu^{\rho}(\alpha_t^*\ps{S})](V)&=&\mu^{\rho}_{|V}\big(\ps{S}\circ \at|_V\big):=\rho_{|V}\circ \mathfrak{S}_{|V}^{-1}(\at^*\ps{S})(V)\\
&=&\rho_{|V}(\hat{P}_{(\at^*\ps{S})(V)})=\rho_{|V}\circ \at\hat{P}_{\ps{S}_V}\\
&=& \at^{-1}\rho_{|V}\circ \mathfrak{S}_{|V}^{-1}(\ps{S}_V)\\
&=&\rho_{|V}\circ \mathfrak{S}_{|V}^{-1}(\ps{S}_V)= [\mu^{\rho}(\ps{S})](V)\,.
\ea
\end{proof}

\begin{proof}[Proof of C2.]
Let us consider the expression 
$\mu^{\rho}(\at^*\ps{S}\wedge \ps{T})$
and evaluate in some context $V_{\alpha}$ (which should be
a $\sigma$-weakly dense $\alpha$-invariant *-sub-algebra of some context $V$). This gives
\ba
[\mu^{\rho}(\at^*\ps{S}\wedge \ps{T})](V_{\alpha})&:=&(\rho_{|V_{\alpha}}\circ \mathfrak{S}_{|V_{\alpha}}^{-1})(\alpha_{t}^*(\ps{S})\wedge \ps{T})(V_{\alpha})\\
&=&\rho_{|V_{\alpha}}(\mathfrak{S}_{|V_{\alpha}}^{-1}\alpha_{t}^*(\ps{S})(V_{\alpha})\wedge \mathfrak{S}_{|V_\alpha}^{-1}\ps{T}(V_{\alpha}))\\
&=&\rho_{|V_{\alpha}}(\alpha_{t}\hat{P}_{\ps{S}_{V_{\alpha}}}\wedge \hat{P}_{\ps{T}_{V_{\alpha}}})\\
&=&\rho_{|V_{\alpha}}( \hat{P}_{\ps{T}_{V_{\alpha}}}\wedge \alpha_{(t+i\beta)}\hat{P}_{\ps{S}_{V_{\alpha}}})\\
&=&[\mu^{\rho}(\ps{T}\wedge \alpha_{(t+i\beta)}^* \ps{S})](V_{\alpha})\,,
\ea
where the fourth equality follows from the second condition of a KMS state on a suitable $\sigma$-weakly dense $\alpha$-invariant *-sub-algebra $V_{\alpha}\subseteq V$.  We thus obtain 
\be
\mu^{\rho}(\alpha_{t}^*\ps{S}\wedge \ps{T})(V_{\alpha})=
\mu^{\rho}(\ps{T}\wedge \alpha_{(t+i\beta)}^* \ps{S})(V_{\alpha}) =  F_{\ps{T},\ps{S}}(t+i \beta) \,.
\ee 
which is precisely the boundary condition on the measure. 
The fact that $[\mu^{\rho}(\ps{S}\wedge \alpha(z)^* \ps{T})](V_{\alpha})$ is analytic in the complex strip and continuous on the boundary follows from  from the definition of $V_{\alpha}$.\footnote{From now on, the analyticity and continuity properties of condition C2 will be immediate. Furthermore, the notation $V_\alpha$ 
will always refer  to a suitable dense *-sub-algebra  of any context $V$
so that the automorphism $\alpha_{t + i \beta}$ finds a sense.}.
\end{proof}
We now want to relate conditions C1 to invariance of the measure with respect to the automorphisms group. To this end we will utilise Lemma 23 in \cite{gelfandII} which shows that $\at\cdot\mu^{\rho}=\mu^{ \at\rho}=\mu^{\rho\circ \at^{-1}}$ is the measure associated to the state
$\at \cdot\rho$. With this in mind, condition C1 becomes

\ba
[\mu^{\rho}(\at^*\ps{S})](V)&=& \mu_{|V}^{\rho}(\ps{S}\circ \at|_V):=\rho_{|V}\circ \mathfrak{S}_{|V}^{-1}(\at^*\ps{S})(V)\\
&=&\rho_{|V}(\hat{P}_{(\at^*\ps{S})(V)})=\rho_{|V}\circ \at\hat{P}_{\ps{S}_V}\\
&=& \at^{-1}\rho_{|V}\circ \mathfrak{S}_{|V}^{-1}\ps{S}_V\\
&=& [\at^{-1} \cdot \mu^{\rho}(\ps{S})](V)\\
&=&[\mu^{\at^{-1}\rho}(\ps{S})](V)\\
&=&[\mu^{\rho}(\ps{S})](V)
\ea
which is clearly satisfied for KMS states since $\at^{-1}\,\rho=\rho$.
\footnote{Note that in the case of von
Neumann algebras, as stated in Definition \ref{def:0kms}. a KMS state $\rho$ is actually normal thus it can be defined in terms of a trace, i.e. $\rho:=\tr(\varrho-)$. Given such a state conditions C1 and C2 are easily derived. In fact the KMS measure condition C1 becomes
\ba
[\mu^{\rho}(\ps{S})](V)&:=&\tr(\varrho\hat{P}_{\ps{S}_V})\crcr
&=&[\mu^{\rho}(\alpha(t)^*\ps{S})](V):=\tr(\varrho\hat{P}_{(\alpha(t)^*\ps{S})_V})\crcr
&=&\tr(\varrho\hat{P}_{\ps{S}_{\alpha(t)V}})\,.
\ea
This is trivially satisfied for a KMS state $\rho$ in fact, we obtain
\be
[\mu^{\rho}(\ps{S})](V):=\tr(\varrho\hat{P}_{\ps{S}_V})=\tr(\alpha_{-t}\varrho\hat{P}_{\ps{S}_V})=
\tr(\varrho \at\hat{P}_{\ps{S}_V})=\tr(\varrho\hat{P}_{\ps{S}_{\at (V)}})=[\mu^{\rho}(\at^*\ps{S})](V)\,.
\ee
On the other hand, given a KMS state $\rho$ condition C2 is shown to hold:
\ba
[\mu^{\rho}(\at^*\ps{S}\wedge \ps{T})](V_{\alpha})
&=&\tr(\varrho \hat{P}_{(\at^*\ps{S})_{V_{\alpha}}}\hat{P}_{\ps{T}_{V_{\alpha}}})=
\tr(\varrho\, \alpha_{t} \hat{P}_{\ps{S}_{V_{\alpha}}}\hat{P}_{\ps{T}_{V_{\alpha}}})\\
&=&\tr(\varrho \hat{P}_{\ps{T}_{V_{\alpha}}}\alpha_{(t + i\beta)} \hat{P}_{\ps{S}_{V_{\alpha}}})=\tr(\varrho \hat{P}_{\ps{T}_{V_{\alpha}}} 
\hat{P}_{(\alpha_{(t + i\beta)}^*\ps{S})_{V_{\alpha}}})
=[\mu^{\rho}(\ps{T}\wedge \alpha_{(t+i\beta)}^* \ps{S})](V_{\alpha})\,.
\nonumber
\ea
The third equality follows from the second KMS condition on $\rho$. The fact that $[\mu^{\rho}(\ps{S}\wedge \alpha(z)^* \ps{T})](V_{\alpha})$ is analytic in the complex strip and continuous on the boundary follows from the trace properties in the case of normal states.
}

If we now consider the topos formulation of quantum theory defined in \cite{andreas5} and take $\mn=\mb(\mh)$, then the automorphisms group $\alpha$ becomes automatically internal and the KMS state has unique solution as Gibbs state. 
Thus we see that in the case of topos quantum theory defined for $\mn=\mb(\mh)$ our definition of general KMS state reduces to that of Gibbs state as it should do.

Let us now go back to Proposition \ref{pro:kms}, conditions C1 and C2 only refer to conditions on the measure associated with the KMS state. However, since there is an injective correspondence between states and truth objects in the topos $\Sets^{\cvn^\op}$, we can translate the above conditions to conditions on the truth object which should represent the topos analogue of the KMS state. 
In the following, we will focus only on condition C1, but the same
analysis can be performed for C2. The main idea is to define a topos
state using the measure $\mu^\rho$ having well defined properties.

As explained in more details in Appendix \ref{app:truth}, when introducing truth objects, one has to extend the topos to the topos of sheaves $\Sh(\mv(\mn)\times (0,1)_L)$. In this setting,  contexts now become pairs $(V, r)$ with $V\in \cvn$ and $r\equiv (0,r)\in (0,1)_L$ (see Appendix \ref{app:truth} for details pertaining to these notations and also \cite{probabilities}). 

\begin{Definition}[KMS condition C1 on truth objects]
\label{def:truth}
Given a truth object $\tob^{\rho}$, for each $(V, r)\in\cvn\times (0,1)_L$, $\tob^{\rho}$ satisfies
a condition C1 if it is defined by all sub-objects of the spectral presheaf, whose measure with respect to $\mu^{\rho}$ is invariant under $\Aut$.  In symbol, that is:
\be
\tob^{\rho}_{(V,r)}=\{\ps{S}\in 
\Sub(\us_{|\dV})|\,\forall V'\subseteq V,\; 
[\mu^{\rho} (\ps{S})](V')=
[\mu^{\rho}(\at ^*(\ps{S}))](V')\geq r, \;
\forall  \at \in \Aut\}\,.
\ee
\end{Definition}
One notices that this definition implies that the objects $\tob^{\rho}$ and $\at^*\tob^{\rho}$ are equivalent 
under a precise sense:
\begin{Definition}[Topos state $\mu$-equivalence]
\label{def:equi1}
Two truth objects $\tob$ and $\tob'$ are said to be $\mu$-equivalent (or equivalent under the measure $\mu$) iff, at each context $(V,r)$, for each $\ps{S}\in \tob_{(V,r)}$ there exists a $\ps{S}'\in\tob'_{(V,r)}$, such that 
$[\mu(\ps{S})](V')=[\mu(\ps{S}')](V')$, for all $V'\subseteq V$. Then, we write $\tob\simeq_{\mu}\tob'$.
\end{Definition}
This is trivially an equivalence relation. Given a truth object $\tob$, we will then denote its $\mu$-equivalence class by $[\tob]_{\mu}$.
Given the condition  C1 for a topos KMS state $\tob^{\rho}$, it follows that for such a state 
$\at^*\tob^{\rho}\simeq_{\mu^{\rho}} \tob^{\rho}$, for all 
$\at\in \Aut$. However, the KMS condition C1 tells us slightly more. 
Indeed, it also specifies which elements have the same measure, namely: $\mu^{\rho}(\ps{S})=\mu^{\rho}(\at^*\ps{S})$. To implement this extra condition, one  defines the notion of a $\mu$-invariant natural transformation.
\begin{Definition}[$\mu$-invariant natural transformation]
\label{def:equi2}
Given a $\mu$-equivalence topos state class $[\tob]_{\mu}$,  a $\mu$-invariant natural transformation $f_{\mu}$ is a natural transformation between a pair of elements $\tob'$ and $\tob''$
 (representatives of $[\tob]_\mu$)  which is defined, for each context $(V,r)$, as follows
\ba
f_{\mu,V}:\tob'_{(V,r)}&\rightarrow& \tob''_{(V,r)}\crcr
\ps{S}&\mapsto&f_{\mu}\ps{S}
\ea
where $f_{\mu}\ps{S}$ is the unique element such that $[\mu(f_{\mu}\ps{S})](V)=[\mu(\ps{S})](V)$. 
\end{Definition}
To show that, indeed, $f_{\mu}$ as defined above is a natural transformation, we must show that given any two contexts $(V,r)$ and $(V', r')$, such that $i:(V,r)\geq(V', r')$ 
(see Appendix \ref{app:truth}, Equation \eqref{vrvr}), the following diagram commutes
\[\xymatrix{
\tob'_{(V,r)}\ar[rr]_{f_{\mu,V}}\ar[dd]_{\tob'(i)}&&\tob''_{(V,r)}\ar[dd]^{\ps{\mathbb{T}}^{''}(i)}\\
&&\\
\tob'_{(V',r')}\ar[rr]^{f_{\mu,V'}}&&\tob''_{(V',r')}\\
}\]
In one direction, we have
\be
\tob''(i)\circ f_{\mu,V}(\ps{S})=\tob''(i)(f_{\mu}\ps{S})=(f_{\mu}\ps{S})_{|\dV'}\,.
\ee
In the other one, we get
\be
f_{\mu,V'}\circ \ps{\mathbb{T}}^{'}(i)(\ps{S})=f_{\mu,V'}(\ps{S})_{|\dV'}=f_{\mu}(\ps{S})_{|\dV'}\,.
\ee
Here, $\ps{S}\in \Sub(\us_{|\dV})$, therefore $\ps{S}=\ps{S}_{\dV}$ which implies that $(f_{\mu}\ps{S})_{|\dV'}=(f_{\mu}\ps{S}_{|\dV})_{|\dV'}=f_{\mu}(\ps{S})_{|\dV'}$.

We now join Definitions \ref{def:equi1} and \ref{def:equi2} in order to define the notion of strongly $\mu$-equivalence as follows:
\begin{Definition}[Topos state strong $\mu$-equivalence]
Two objects $\ps{\mathbb{T}}^{'}$ and $\ps{\mathbb{T}}$, are strongly $\mu$-equivalent iff they are $\mu$-equivalent and there exists a $\mu$-invariant natural transformation between them. We will denote a strong $\mu$-equivalence class as  $[[\tob]]_{\mu}$.
\end{Definition}

Using this notion, we can write the first KMS condition as:
\begin{Definition}[KMS condition C1 on strong $\mu$-classes]
A truth object $\tob^{\rho}$ satisfies a KMS condition C1
at the strong $\mu$-equivalent class level if, for all $\at\in \Aut$,
\be
[[\tob^{\rho}]]_{\mu^{\rho}}=[[\at^*\,\tob^{\rho}]]_{\mu^{\rho}}\,,
\ee
with $\mu$-invariant natural transformation defined as
\ba
\tob^{\rho}_{(V,r)}&\rightarrow& \at^*\,\tob^{\rho}_{(V,r)}\crcr
\ps{S}&\mapsto&\at^*\,(\ps{S})\,.
\ea
\end{Definition}
As a consequence of the above, for each presheaf $\tob^{\rho}$, all the twisted presheaves derived from the morphism $\at^*$ for all $\at\in \Aut$ are $\mu^{\rho}$-equivalent. In other words, they select 
sub-objects of  the respective state spaces $\us$ which have exactly the same measure.
We can finally read off the full KMS conditions.
\begin{Definition}\label{def:kmstopos}[Topos state KMS conditions]
Given a KMS state $\rho$, the truth object $\ps{\mathbb{T}}^{\rho}$ is the topos analogue of the KMS state if the following conditions are satisfied
\begin{enumerate}
\item For all $\at\in \Aut$,
\be
[[\ps{\mathbb{T}}^{\rho}]]_{\mu^{\rho}}=[[\at^*\ps{\mathbb{T}}^{\rho}]]_{\mu^{\rho}}\,,
\ee
with the $\mu$-invariant natural transformation defined  as
\ba
\tob^{\rho}_{(V,r)}&\rightarrow& \at^*\tob^{\rho}_{(V,r)}\crcr
\ps{S}&\mapsto&\at^*(\ps{S})
\ea
 \item For each context $(V, r)\in (\mv(\mn)\times (0,1)_L)$, given any two elements $\ps{S},\ps{T}\in\tob^{\rho}_{(V, r)}$, the function 
   \be
F_{\ps{T},\ps{S}}(t) = \mu^{\rho} (\ps{T}\wedge\alpha (t)^*\ps{S})](V)
\ee

for all $t\in \Rl$ admits an extension $F_{\ps{T},\ps{S}}(z)$ analytic in the complex strip  $D_\beta = \{z = (t + i y)|\; t\in \Rl, \, y \in ]0,\beta[\}$
and continuous at the boundary of $D_\beta$ such that
\be
F_{\ps{T},\ps{S}}(t+ i\beta) = 
[\mu^{\rho}(\at^*\ps{S}\wedge\ps{T})](V)
\geq r\,,
\ee
\end{enumerate}
\end{Definition}
Alternatively condition 2) above can be stated as follows
\begin{enumerate}
\item [2')] For each context $(V_{\alpha}, r)\in (\mv(\mn)\times (0,1)_L)$, given any two elements $\ps{S},\ps{T}\in\tob^{\rho}_{(V_{\alpha}, r)}$ the function 
 \be
F_{\ps{T},\ps{S}}(t+ i\beta) := [\mu^{\rho}(\alpha(t)^*\ps{T}\wedge\ps{S})](V'_{\alpha})=[\mu^{\rho}(\ps{S}\wedge\alpha(t+i\beta)^*\ps{T})](V'_{\alpha})\geq r\,.
\ee
where $V'_{\alpha}\subseteq V_{\alpha}$. 
$F_{\ps{T},\ps{S}}(t+ i\beta)$ is analytic in the complex strip $D_\beta = \{z = (t + i y)|\; t\in \Rl, \, y \in ]0,\beta[\}$ and continuous at the boundary of $D_\beta$.
\end{enumerate}

\subsection{Deriving the canonical KMS state from the topos KMS state}
\label{sect:topotocan}

As noticed in the previous analysis, the definition of the
 topos KMS state follows from the very existence of an initial KMS state
(that we call canonical). 
We are now interested in proving the reciprocal, i.e. starting with the topos KMS state definition and derive the canonical KMS state. To this end, we will use the fact that there is a 1:1 correspondence between states $\rho$ and truth objects $\tob^{\rho}$. Since $\ps{\t}^{\rho}$ is defined through the measure $\mu^{\rho}$ on $\us$, we will first 
concentrate on the link between the topos KMS measure and the
canonical  KMS state. 

The following statement holds:

\begin{Theorem}\label{the:derivekms}
Consider a  measure $\mu$ and a one parameter automorphism group $\Aut=\{\alpha(t)|\,t\in\Rl\}$ on $\mn$, whose action on the category $\cV(\mn)$ is
\ba
\Aut\times \cV(\mn)&\rightarrow& \cV(\mn)\crcr
(\alpha(t), V)&\mapsto&[\alpha(t)](V)
\ea
and such that $\Aut$ can be extended to a complex strip $\{\alpha(t+i\gamma)|\,t\in\Rl,\, \gamma\in[0,\beta]\}$. 
 If $\mu$ satisfies the following conditions
\be\label{equ:1cond}
\widetilde{C1}: \qquad \mu(\ps{S})=\mu(\alpha(t)^*\ps{S})
\ee
and $\widetilde{C2}:$ for all $\ps{S}$, $\ps{T}$ in $\Sub$ and $V_{\alpha}\in \mv(\mn)$, if $F_{\ps{T},\ps{S}}(z): = [\mu(\ps{T}\wedge\alpha(z)^*\ps{S})](V_{\alpha})$ is analytic in the complex open 
strip $D_\beta$ and continuous at the boundary of $D_\beta$ such that
 \be\label{equ:2cond}
F_{\ps{T},\ps{S}} (t+ i \beta): =[\mu(\ps{T}\wedge\alpha(t+ i \beta)^*\ps{S})] (V_{\alpha})=
[\mu(\alpha(t)^*\ps{S}\wedge\ps{T})](V_{\alpha})\,,
\ee
then the state $\rho^{\mu}$ associated to such a topos measure is the KMS state associated to the algebra $\mn$.
\end{Theorem}
\begin{proof}
In order to define a state from a measure, we apply Theorem \ref{theo:measure} (see Appendix \ref{subsect:statfromes}). In particular, \eqref{equ:defmu} which can be written 
 \be
m(\hat{P}):=[\mu(\ps{S}_{\hat P})](V)
\ee
will be useful. Since our measure satisfies conditions \eqref{equ:1cond} and \eqref{equ:2cond}, it can be shown that
 \be
m(\hat{P}):=[\mu(\ps{S}_{\hat P})](V)=[\mu(\at^*\ps{S}_{\hat P})](V)=[\mu(\ps{S}_{\hat P})](\at(V))=m(\at\hat{P})\,.
\ee
Because $\mu$ is a finitely additive probability measure, by construction $m:\mP(\mn)\rightarrow [0,1]$ will be finitely additive on the projections in $\mn$, hence we can apply the generalised version of Gleason's theorem\footnote{Gleason Theorem tells us that the only possible probability measures on Hilbert spaces of dimension at least 3 are measures of the form $\mu(P)=\Tr(\rho P)$, where $\rho$ is a positive semidefinite self-adjoint operator of unit trace. This theorem was extended to a von-Neumann algebra $\mn$ in \cite{Gleason} where the author shows that, provided $\mn$ contains no direct summand of type $I_2$, then every finitely additive probability measure on $\mP(\mn)$ can be uniquely extended to a normal state on $\mn$. The general form of Gleason theorem is as follows
\begin{Theorem}
Assume that $dim(\mh)\geq 3)$ and let $\mu$ be a $\sigma$-additive probability measure on $P(\mb(\mh))$ then the following three statement hold
\begin{enumerate}
\item $\mu$ is completely additive.
\item $\mu$ has support.
\item There exists a positive operator $x\in\mb(\mh)$ of trace class such that $tr(x)=1$ and $\mu(e)=tr(xe)$ for $e\in P(\mb(\mh))$
\end{enumerate}
\end{Theorem}
}\cite{maeda} and obtain the unique state $\rho_m:\mn\rightarrow \Cl$ such that $\rho_m|_{\mP(\mn)}=m$. Hence

\be
m(\hat{P})=\rho_m|_{\mP(\mn)}(\hat{P})=m(\at\hat{P})=\rho_m|_{\mP(\mn)}(\at\hat{P})\,.
\ee
Similarly, for the second condition, we have
 \be
m(\at\hat{P}\hat{R})=
[\mu(\at^*\ps{S}_{\hat P}\wedge
\ps{T}_{\hat R})](V_{\alpha})=[\mu(\ps{T}_{\hat R}\wedge\alpha_{t+i\beta}^*\ps{S}_{\hat P})](V_{\alpha})
=m(\hat{R}\alpha_{t+i\beta}\hat{P})
\ee
which translates to
  \be
\rho_m|_{\mP(\mn_\alpha)}(\at\hat{P}\hat{R})=\rho_m|_{\mP(\mn_\alpha)}( \hat{R}\alpha_{t+i\beta}\hat{P})\,.
\ee 
where $\mn_{\alpha}$ is a $\sigma$-weakly dense $\alpha$-invariant $*$-sub-algebra of $\mn$ as defined in Definition \ref{def:1kms}.
\end{proof}

We are in position to derive the canonical KMS state from the topos KMS state. 
\begin{Theorem}
Given a topos state $\tob^{\rho}$ defined via the topos
measure $\mu^\rho$ satisfying $\widetilde{C1}$ and $\widetilde{C2}$, the unique state $\tilde\rho$ corresponding to $\mu^\rho$ is a KMS state in the canonical sense.
\end{Theorem} 
\begin{proof}
The first condition on $\tob^{\rho}$ is 
\be
[[\ps{\mathbb{T}}^{\rho}]]_{\mu^{\rho}}=[[\at^*\ps{\mathbb{T}}^{\rho}]]_{\mu^{\rho}}\,,
\ee
where the $\mu$-invariant natural transformation is defined for each $\at$ as
\ba
\tob^{\rho}&\rightarrow& \at^*\,\tob^{\rho}\\
\ps{S}&\mapsto&\at^*(\ps{S})
\ea
Given two elements $\tob^{\rho}$ and $\at^*\,\tob^{\rho}$ in $[[\ps{\mathbb{T}}^{\rho}]]_{\mu^{\rho}}$, for each context $(V, r)$ the local components are, respectively,
\be
\tob^{\rho}_{(V,r)}=\{\ps{S}\in \Sub(\us_{|\dV})|\;\forall V'\subseteq V,\, [\mu^{\rho}(\ps{S})](V)\geq r\}
\ee
and
\be
\at^*\,\tob^{\rho}_{(V,r)}=\{\at^*\ps{S}\in \Sub(\at^*\us_{|\dV})|\;\forall V'\subseteq V,\, [\mu^{\rho}(\at^*\ps{S})](V)\geq r\}\,.
\ee
From the definition of the $\mu$-invariant natural transformation, we now know how to identify objects with the same measure: 
\be\label{equ:condition1}
 [\mu^{\rho}(\at^*\ps{S})](V)= [\mu^{\rho}(\ps{S})](V)\,.
\ee
The second condition on the topos state tells us that, for all $V_{\alpha}\in\cvn$,
\be\label{equ:condition2}
[\mu^{\rho}(\at^*\ps{S}\wedge\ps{T})](V_{\alpha})=[\mu^{\rho}(\ps{T}\wedge\alpha_{t+i\beta}^*\ps{S})](V_{\alpha})
\ee
Applying Theorem \ref{the:derivekms} to conditions \eqref{equ:condition2} and \eqref{equ:condition1}, one infers that $\rho$ is indeed a KMS state.
\end{proof}

\subsection{Expectation Values}
\label{subsect:inner}

In this section we would like to analyse the consequences of the measure KMS condition on expectation values as defined in topos quantum theory. In particular, we seek an analogue condition for $\langle \phi, \hat{A}\rangle=\langle \phi, \at\hat{A}\rangle$ in the topos formulation which is compatible with the first condition on the measure. A primary notion towards this 
generalization would be the one of expectation value in the topos formalism
\cite{probabilities}. Let us review quickly the main features
of that latter notion.

Consider the measure $\mu^{\rho}$ associated to a state $\rho$.
According to \eqref{ali:measure}, the following quantity  
\be
\mu^{\rho}(\ps{\delta(\hat{P})}):\cvn\rightarrow [0,1]
\ee
is an OR function such that, for all $V\in \cvn$, $[\mu^{\rho}(\ps{\delta(\hat{P})})](V):=\rho( \dase(\hat{P})_{V})$
(daseinisation $\delta$ and $\dase$ have been defined in Appendix
\ref{subsect:proj}). Thus, if $V$ is such that $\hat{P}\in V$, then $\dase(\hat{P})_V=\hat{P}$ and $[\mu^{\rho}(\ps{\delta(\hat{P})})](V)=\tr(\varrho\,\hat{P})$. For all other contexts $V'$, such that $\hat{P}\notin V'$, then $\dase(\hat{P})_{V'}>\hat{P}$, therefore $[\mu^{\rho}(\ps{\delta(\hat{P})})](V')>\rho(\hat{P})$.
It follows that the expectation value $E(\hat{P},\rho)$ is the minimum of the function $\mu^{\rho}(\ps{\delta(\hat{P})}):\cV(\mathcal{N})\rightarrow [0,1]$, i.e.
\be
E(\hat{P},\rho)=\langle\rho\,;\hat{P}\rangle= \rho(\hat{P})=\min_{V\in\cV(\mathcal{N})}[\mu^{\rho}(\ps{\delta(\hat{P})})](V)\,.
\label{eprho}
\ee 
Consider a self-adjoint operator $\hat{A}\in \mathcal{N}$ which can be written in the form 
\be
\hat{A}=\sum_{i}^na_i\hat{P}_i\,, \quad a_i \in \Rl\,,
\ee
for pairwise orthogonal projections $\hat{P}_i\in \mP(\mathcal{N})$. Given a state $\rho$ and using \eqref{eprho}, the expectation value of $\hat{A}$ is then
\be
E(\hat{A}, \rho)=\langle\rho\,;\hat{A}\rangle=\sum_{i}^na_i
\langle\rho\,;\hat{P}_i \rangle =\sum_i^na_i \min_{V\in\cvn}[\mu^{\rho}(\ps{\delta(\hat{P}_i)})](V)\,.
\ee
It follows that the topos analogue of the KMS condition $\langle \phi, \hat{A}\rangle=\langle \phi, \at\hat{A}\rangle$ is
\ba
\label{ali:tkms}
E(\hat{A},\rho)&=&E(\at\hat{A}, \rho)\,,\crcr
\sum_i^na_i \min_{V\in\cvn}[\mu^{\rho}(\ps{\delta(\hat{P}_i)})](V) &=&\sum_i^na_i \min_{V\in\cvn}[\mu^{\rho}(\ps{\delta(\at\hat{P}_i)})](V)\,.
\ea
We can relate this condition to the condition on the measure defined in Proposition \ref{pro:kms}. To this end, we first analyze the effect that the group elements $\at$ have on daseinisation. In particular, we consider
 \ba
\at[\dase(\hat{P})(V)]&=&\at\big(\bigwedge\{\hat{R}\in \mP(V)|\,\hat{R}\geq \hat{P}\}\big)\crcr
&=&\bigwedge\{\at \hat{R}\in \mP(\at V)|\,\at \hat{R}\geq \at \hat{P}\}\nonumber\\
&=&\dase(\at \hat{P})(\at V)\;.\nonumber
\ea 
From this, it follows that, given a KMS state $\rho$,
\ba
[\mu^{\rho}(\at^*\ps{\delta(\hat{P})})](V)&=&\rho\big(\at^*\ps{\delta(\hat{P})}(V)\big)\crcr
&=&\rho(\ps{\delta(\hat{P})}(\at V))=\at\, \rho(\ps{\delta(\hat{P})}(\at V))\crcr
&=&\rho\circ \at^{-1}(\ps{\delta(\hat{P})}(\at V))=\rho(\ps{\delta(\at^{-1}\hat{P})}(V))\crcr
&=&[\mu^{\rho}(\ps{\delta(\at^{-1}\hat{P})})](V)\,,
\ea
where the third and fourth equality follow from the fact that $\rho$ is a KMS state\footnote{ It is worth noting that the last equation in \eqref{ali:tkms} can also be written as $\sum_i^na_i \min_{V\in\cvn}\mu^{\rho}(\ps{\delta(\at \hat{P}_i)})(\at V)$, 
since $\at:\cV(\mn)\rightarrow\cV(\mn)$. 
Writing the equation in such a way and using the fact that $\dase(\at\hat{P})(\at V)=\at[\dase(\hat{P})(V)]$, we obtain that $[\mu^{\rho}(\at^*\ps{\delta(\hat{P})})](V)=[\mu^{\rho}(\ps{\delta(\at \hat{P})})](\at V)$, 
hence $[\mu^{\rho}(\ps{\delta(\at \hat{P})})](\at V)=[\mu^{\rho}(\at^*(\ps{\delta(\hat{P})})](V)$.}. 
Thus
\be
[\mu^{\rho}(\ps{\delta(\at \hat{P})})](V)=[\mu^{\rho}(\at^{-1})^*(\ps{\delta(\hat{P})})](V)\,.
\ee
Putting this result in \eqref{ali:tkms}, we obtain the topos analogue of the first KMS condition to be
\be
\sum_i^na_i \min_{V\in\cvn}[\mu^{\rho}(\ps{\delta(\hat{P}_i)})](V)=\sum_i^na_i \min_{V\in\cvn}[\mu^{\rho}(\at^{-1})^*\ps{\delta(\hat{P}_i)})](V)\,.
\ee
This is in accordance with the result obtained above, namely for KMS states $\rho$, the measure is invariant under transformation of the automorphism group. 

The second condition of Proposition \ref{pro:kms} can be related 
to the second condition on the KMS state. 
Considering the product of two operators $\hat{A}\hat{B}$ where $A, B$ are in a norm dense $\alpha$-invariant *-subalgebra $V_{\alpha}$ of $\mn$. Each such operator can be written as $\hat{A}=\sum_i^na_i\hat{P}_i$ and $\hat{B}=\sum_j^mb_j\hat{P}_j$.
Since $\hat{P}_i\hat{P}_j=\hat{P}_i\wedge\hat{P}_j$, in the topos 
framework, the relevant objects will be of the form 
$\dase(\hat{P}_i\wedge\hat{P}_j)(V)$, for some context $V$. 
In general, it is true that $\dase(\hat{P}_i\wedge\hat{P}_j)(V)\leq \dase(\hat{P}_i)(V)\wedge\dase(\hat{P}_j)(V)$, however, for those contexts $V$ for which $\hat{P}_i\wedge \hat{P}_j,\hat{P}_i, \hat{P}_j\in V$, then the same relation trivially reduces to $\dase(\hat{P}_i\wedge\hat{P}_j)(V)=\hat{P}_i\wedge\hat{P}_j$. Clearly such contexts will be of the form $V_{\alpha}$. 
Moreover, if $\hat{P}_i, \hat{P}_j\in V$, then we have
\be
\dase(\hat{P}_i\wedge\hat{P}_j)(V)\geq \dase(\hat{P}_i)(V)\wedge\dase(\hat{P}_j)(V)=\hat{P}_i\wedge\hat{P}_j\,.
\ee
On the other hand, if $\hat{P}_i\wedge \hat{P}_j\in V,$ then
\be
\dase(\hat{P}_i\wedge\hat{P}_j)(V)=\hat{P}_i\wedge\hat{P}_j\leq \dase(\hat{P}_i)(V)\wedge\dase(\hat{P}_j)(V)\,.
\ee
Hence, computing the minimum with respect to all contexts $V\in\cvn$, the two expressions yield the same result.
Thus, the expectation value of $\hat{P}_i\wedge\hat{P}_j$ is
 \be
E(\hat{P}_i\wedge\hat{P}_j, \rho)=\min_{V\in\cvn}
[\mu^{\rho}(\ps{\delta(\hat{P}_i\wedge \hat{P}_j}))](V) =\min_{V\in\cvn}[\mu^{\rho}(\ps{\delta(\hat{P}_i)}\wedge \ps{\delta(\hat{P}_j)})](V)=\min_{V_{\alpha}\in\cvn}[\mu^{\rho}(\ps{\delta(\hat{P}_i)}\wedge \ps{\delta(\hat{P}_j)})](V_{\alpha})\,.
\ee
where the last equality follows since we have assumed that 
$\hat{A}, \hat{B}\in V_{\alpha}$ .
In this situation, we obtain
\be
E(\hat{A}\hat{B}, \rho)=\sum_i^n\sum_j^m a_ib_j\min_{V\in\cV(\mathcal{N})}\mu^{\rho}(\ps{\delta(\hat{P}_i)}\wedge \ps{\delta(\hat{P}_j)})(V)\,.
\ee
We would like to mimicking the KMS condition $\langle\rho, \hat{A}\alpha_{t+i\beta}\hat{B}\rangle=\langle\rho, \at\hat{B}\hat{A}\rangle$ for $A, B$ in some $V_{\alpha}$. However, since $\langle\rho, \hat{A}\alpha_{t+i\beta}\hat{B}\rangle=\langle\rho, \hat{A}\alpha_{i\beta}\at\hat{B}\rangle$, then condition  
$\langle\rho, \hat{A}\alpha_{t+i\beta}\hat{B}\rangle=\langle\rho, \at\hat{B}\hat{A}\rangle$ simply reduces to $\langle\rho, \hat{A}\alpha_{i\beta}\hat{B}\rangle=\langle\rho,\hat{B}\hat{A}\rangle$. In topos language this gets translated, for all $V_{\alpha}\in\mv(\mn)$, to
\be
\sum_i^n\sum_j^m a_ib_j\min_{V_{\alpha}\in\cvn}[\mu^{\rho}(\ps{\delta(\hat{P}_i)}\wedge \ps{\delta(\alpha_{i\beta}\hat{P}_j)})](V_{\alpha})=
\sum_i^n\sum_j^m a_ib_j\min_{V_{\alpha}\in\cvn}[\mu^{\rho}(\ps{\delta(\hat{P}_j)}\wedge\ps{\delta(\hat{P}_i)} )](V_{\alpha})
\ee
which is equivalent to
 \be
\sum_i^n\sum_j^m a_ib_j\min_{V_{\alpha}\in\cV(\mathcal{N})}\mu^{\rho}(\ps{\delta(\hat{P}_i)}\wedge \alpha_{-i\beta}^*\ps{\delta(\hat{P}_j)})(V_{\alpha})=\sum_i^n\sum_j^m a_ib_j\min_{V_{\alpha}\in\cV(\mathcal{N})}\mu^{\rho}(\ps{\delta( \hat{P}_j)}\wedge \ps{\delta(\hat{P}_i)})(V_{\alpha})\,.
\ee 
If we consider for a moment two general sub-objects of the state space: $\ps{S},\ps{T}\subseteq \us$, then the above condition translates to the following 
\be\label{equ:8.44}
[\mu^{\rho}(\ps{S}\wedge (\alpha_{-i\beta}^*\ps{T})](V_{\alpha})=[\mu^{\rho}(\ps{T}\wedge\ps{S})](V_{\alpha})
\ee
Replacing $\alpha_{-i\beta}$ by $\alpha_{i\beta}$ equation \eqref{equ:8.44} becomes 
\be
\mu^{\rho}(\ps{S}\wedge \alpha_{i\beta}^*\ps{T})(V_{\alpha})=\mu^{\rho}(\ps{T}\wedge\ps{S})(V_{\alpha})
\ee
which is the topos analogue of the second KMS conditions: 
$\langle\rho, \hat{A}\,\alpha_{i\beta}\hat{B}\rangle=\langle\rho, \hat{B}\hat{A}\rangle$ and is in agreement with Proposition \ref{pro:kms}.

\subsection{Consequences on truth values}
\label{subsect:cons}
We now discuss the consequences of the KMS condition at the topos level. Given $\tob^{\rho}$, we compute, 
for all contexts $(V,r)\in \cV(\mn)\times (0,1)_L$,
\be
v(\ps{\delta\hat{P}}\in\ps{\mathbb{T}}^{\rho})(V,r)=\{(V',r)\leq (V,r)|\,\mu^{\rho}(\ps{\delta(\hat{P})})(V')\geq r\}\,.
\ee
Applying a transformation by $\at\in \Autr$ on this set, 
we have
\ba
\at\, v(\ps{\delta\hat{P}}\in\tob^{\rho})(V,r)&=&\at\{(V',r)\leq (V,r)|\,\mu^{\rho}(\ps{\delta(\hat{P})})(V')\geq r\}\crcr
&=&\{(\at V',r')\leq (\at V,r)|\,
[\mu^{\rho}(\ps{\delta(\hat{P})})](V')\geq r'\}\crcr
&=&\{(\at V',r')\leq (\at V,r)|\,
[\mu^{\rho}(\ps{\delta(\at\hat{P})})](\at V')\geq r'\}\crcr
&=&\{(\at V',r')\leq (\at V,r)|\,
[\mu^{\rho}(\at^*\ps{\delta(\hat{P})})](V)\geq r'\}\crcr
&=&v(\at^*\ps{\delta(\hat{P})}\in\tob^{\rho})(\at V,r)\,.
\ea
However $v\big(\at^*\ps{\delta(\hat{P})}\in\tob^{\rho}\big)(\at V,r)=v\big(\at^*\ps{\delta(\hat{P})}\in\tob^{\rho}(\at V,r)\big)=v\big(\at^*\ps{\delta(\hat{P})}\in\at^*\ps{\t}^{\rho}(V,r)\big) =v\big(\at^*\ps{\delta(\hat{P})}\in\at^*\ps{\t}^{\rho}\big)(V,r)$. Since $\mu^{\rho}(\at^*\ps{\delta(\hat{P})})(V)=\mu^{\rho}(\ps{\delta(\hat{P})})(V)$,  we get 
\ba
v(\at^*\ps{\delta(\hat{P})}\in\ps{\t}^{\rho})(V,r)&=&\{(\at V',r')\leq (\at V,r)|\,\mu^{\rho}(\at^*\ps{\delta(\hat{P})})(V)\geq r'\}\crcr
&=&\{(\at V',r')\leq \langle\at V,r)|\,\mu^{\rho}(\ps{\delta(\hat{P})})(V)\geq r'\}\crcr
&=&v(\ps{\delta(\hat{P})}\in\ps{\mathbb{T}}^{\rho})(V,r)\,.
\ea
Combining the two results, it follows that, given a KMS state $\tob^{\rho}$, then, for each context $(V,r)\in \cvn\times (0,1)_L$, the truth values of a given proposition $\ps{\delta(\hat{P})}$ are invariant under $\Autr$, i.e.
\be
v(\ps{\delta\hat{P}}\in\tob^{\rho})(V,r)=\at\,v(\ps{\delta\hat{P}}\in\tob^{\rho})(V,r)\,.
\ee

\subsection{Application: Topos KMS state on $\Cl^{3}$}
\label{subsect:applic}

In this section, we give an example of the topos analogue of the KMS state with the external condition. Our starting point will be the three dimensional Hilbert space $\Cl^3$ with orthonormal basis denoted by $\{|i\rangle\}$, $i=1,2,3$. In this context, $\mb(\mh)=M_3(\Cl)$, while our KMS state density\footnote{Note here our state $\rho$ is normal.} is $\varrho:=\sum_{i=1}^3a_i\hat{P}_i$ where $\hat{P}_i=|i\rangle\langle i|$ are the projection operators spanning $\Cl^3$
and $\sum_{i=1}^3 a_i =1$, $a_i\in [0,1]$. The associated state $\rho$ is $\rho:=\sum_{i=1}^3a_i^{\frac{1}{2}}\hat{P}_i$ which is is cyclic and separating. In fact given any $\hat{X}\in M_3(\Cl)$, if $X$ is orthogonal to all $\hat{A}\rho$ ($\hat{A}\in M_3(\Cl)$ then 
\be
\tr(\hat{X}^*\hat{A}\rho)=\sum_{i=1}^3a_i^{\frac{1}{2}}\langle i|\hat{X}^*\hat{A}i\rangle=0\,, \quad \forall \hat{A}\in  M_3(\Cl)\,.
\ee
In particular taking $A=\hat{P}_i$, then it follows that $\langle i|\hat{X}^*i\rangle=0$ thus $\hat{X}=0$. Similarly one can prove that $\rho$ is cyclic\footnote{Note that it is a standard result that given a faithful normal state on a von-Neumann algebra, the corresponding state is cyclic and separating.}. 
The automorphisms considered are of the form 
$\alpha_\rho(t)A = e^{it \hat H}A e^{-it \hat H}$, 
for all $A\in \mb(\mh)$,
where $\varrho=e^{-\beta \hat H}$, and $\hat H=(-1/\beta)\sum_{i=1}^3 \ln a_i \hat{P}_i$.

For simplicity, we will only consider a particular abelian von Neumann sub-algebra of $\mb(\mh)$ and restrict our calculations of $\tob^{\rho}$ to that sub-algebra. 
Let   $\hat{P}_{12}$ be the projection operator defined by
$\frac{1}{2}(|1\rangle+|2\rangle)(\langle2|+\langle 1|)$.
The sub-algebra we choose is $V={\rm lin}_{\Cl}(\hat{P}_{12},\hat{1}- \hat{P}_{12})$ whose elements are all matrix of the form
\[\begin{pmatrix} a& b& 0\\	
b&a&0 \\
0&0&c	
  \end{pmatrix}
  \]
 Given the simplicity of such an algebra, we have that $\dV=V$. It is straightforward to see that such an algebra $V$ is its own commutant
and therefore $V''=V$ is a von Neumann algebra. 
Within such a setting, we first compute the spectral presheaf for  the context $V$. Since such an algebra has only two generators, the spectral presheaf at this context only contains two elements, namely
\be
\us_V=\{\lambda_1, \lambda_2\}\;\;\text{ such that }\;\; \lambda_1(\hat{P}_{12})=\lambda_2\big(\hat{1}-\hat{P}_{12}\big)=1\,;\;\quad \lambda_1\big(\hat{1}-\hat{P}_{12}\big)=\lambda_1(\hat{P}_{12})=0\,.
\ee
The next object that we need to determine is $\Sub(\us)_V$. This collection contains only three clopen sub-sets, namely
\be\label{equ:subobjects}
\ps{S_1}_V:=S_1=\{\lambda_1\}\,,\;\quad \ps{S_2}_V:=S_2=\{\lambda_2\}\,,\;\quad \ps{S_{12}}_V:=S_{12}=\{\lambda_1, \lambda_2\}\,.
\ee
Applying the inverse image of the map $\mathfrak{S}:\mathcal{P}(V)\rightarrow \Sub(\us)_V$ defined in \eqref{equ:smap}, we obtain
\be\label{eq:proj}
\mathfrak{S}(S_1)=\hat{P}_{12}\,,\;\quad \mathfrak{S}(S_2)=\hat{1}-\hat{P}_{12}\,,\;\quad \mathfrak{S}(S_{12})=\hat{1}\,.
\ee
We are interested in computing the topos truth object associated to the KMS state $\rho$. This will be the unique (up to the $\mu$-equivalence) truth object $\tob^{\rho}$ satisfying conditions 1 and 2 in Definition \ref{def:kmstopos} for automorphism $\alpha_\rho(t)$ given above. 
Having such a $\tob^{\rho,r}$, it will be a simple task to built
the equivalence class $[[\tob^\rho]]_{\mu^\rho}$ associated to the truth object in the sense we prescribed above. 
We then consider a general $r\in [0,1]$, such that at context $V$ we obtain
\be
\ps{\mathbb{T}}^{\rho, r}_V=\{\ps{S}\in \Sub(\us_{\dV})|\;\forall V_i\subseteq V\,,\; \; [\mu^{\rho}(\ps{S})](V_i)\geq r\}\,.
\ee
It can be checked that conditions C1 and C2 of the measure 
$\mu^\rho$ are satisfied. We recall that the measure on an element $\ps{S_1}\in \ps{\mathbb{T}}_V^{\rho, r}$, is defined as $[\mu^{\rho}(\ps{S_1})](V):=\tr(\varrho \hat{P}_{\ps{S_1}_V})$ for each $V$. Given such an expression for the measure it is straightforward to show that indeed conditions C1 and C2 are satisfied for $\varrho=\sum_{i=1}^3a_i\hat{P}_i$ and the clopen sub-objects with associated projection operators defined in \eqref{equ:subobjects} and \eqref{eq:proj} respectively.

To determine which clopen sub-objects belong to $\tob^{\rho, r}_V$, we need to compute $\mu^{\rho}(\ps{S})(V_i)$ for each of the clopen sub-objects in \eqref{equ:subobjects}. This computation will provide us with conditions to be satisfied in order to determine if a clopen sub-object belongs to the KMS truth object or not. We do a case study:
\begin{enumerate}
\item Clopen sub-object 
$\ps{S_1}$:
\be
\mu^{\rho}(\ps{S_1})(V)=\tr(\varrho\, \hat{P}_{12})=\frac{1}{2}(a_1+a_2)\,.
\ee
Therefore, the condition obtained is 
\be
\text{if } \quad \frac{1}{2}(a_1+a_2) \geq r \,,\quad \text{ then } \quad  \ps{S_1}\in\ps{\mathbb{T}}^{\rho, r}_V\,.
\ee
\item  Clopen sub-object 
$\ps{S_2}$:
\be
\mu^{\rho}(\ps{S_2})(V)=\tr\big(\varrho\, (\hat{1}-\hat{P}_{12})\big)=1-\frac{1}{2}(a_1+a_2)\,.
\ee
Therefore the condition, in this case, is
\be
\text{if }\quad 1- \frac{1}{2}(a_1+a_2) \geq r \,,\quad \text{ then } \quad  \ps{S_2}\in\ps{\mathbb{T}}^{\rho, r}_V
\ee
\item  Clopen sub-object 
$\ps{S_{12}}$:
\be
\mu^{\rho}(\ps{S_{12}})(V)=\tr\big(\varrho \,\hat{1}\big)=1\,.
\ee
Therefore, it becomes obvious that 
\be
\ps{S_{12}}\in\ps{\mathbb{T}}^{\rho, r}_V\,.
\ee
\end{enumerate}

\section{Internal KMS condition}
\label{sect:inter}

Although the topos analogue of the KMS state reproduces the ordinary 
canonical properties, its definition still requires a collection of 
$\mu$-equivalent twisted presheaves. Nevertheless, it is possible to avoid twisted presheaves. This can indeed be done by 
{\it internalizing} the group $\Autr$ (henceforth, for simplicity denoted by $H$) and defining presheaves in terms of it. The tools used 
in order to  achieve this were introduced in \cite{group}. The first step is to define the topos analogue of $H$. 

For simplicity of exposition all properties induced by the external KMS condition in the previous section will not be treated in detail in the present section. We emphasize, however, that they should admit an analogue 
formulation in the internal condition setting.
 
It should be also noted that the notion of {\it external} and {\it internal} topos properties used in the present paper should not be confused with the notion of external and internal description as used in  \cite{comparison}, \cite{bass1} and  \cite{ bass}.

\subsection{The automorphism group}
\label{subsect:auto}

The first step is to define the topos analogue of the automorphisms group $H=\{\at |\,t\in \Rl\}$ of the algebra $\mathcal{N}$. For each element, in this group, we obtain the induced geometric morphism
\ba
\at ^*:\Shvn &\rightarrow& \Shvn \crcr
\ps{S}&\mapsto&\at ^*\ps{S}
\ea
such that $\at ^*\,\ps{S}(V):=\ps{S}(\at V)$.
This action, however, gives rise to twisted presheaves \cite{andreas5}. In order to avoid this feature, the same recipe as performed in \cite{group} can be used. We start with the base category $\cvnf$ which is fixed, i.e.  the group $H$ is not allowed to act on it. Then, we define the {\it internal} group $\ps{H}$ over this new base category as follows:

\begin{Definition}
The internal group $\ps{H}$ is the presheaf defined on
\begin{enumerate}
\item Objects: for each $V\in \cvnf$,  $\ps{H}_V=H$;
\item Morphisms are simply identity maps. 
\end{enumerate}
\end{Definition}
It is straightforward to realize that the global sections of this presheaf reproduce the group, i.e.  $\Gamma(\ps{H})=H$. Next, the fixed point group presheaf  $\ps{H_F}$ can be defined as:
\begin{Definition}
The presheaf $\ps{H_F}$ is defined on
\begin{enumerate}
\item Objects: for each $V\in \cvnf$,  $\ps{H_F}_V = H_{FV}$
subgroup of $H$ called  the fixed point group of $V$ i.e. $H_{FV}=\{\alpha\in\ps{H}|\forall a \in V\alpha a = a\}$.
\item Morphisms: given a morphism $i_{V'V}: V'\subseteq V$, the
corresponding morphism is
\ba
\ps{H_F}(i_{V'V}): (\ps{H_F}_V=H_{FV}) &\to& (\ps{H_F}_{V'}=H_{FV'}) \cr\cr
\alpha &\mapsto&\ps{H_F}(i_{V'V})[\alpha ] = \alpha|_{V'}
\ea
viz. $\ps{H_F}_{V}\rightarrow \ps{H_F}_{V'}$ is given in terms of subgroup inclusion, i.e. $\ps{H_F}_{V}\subseteq \ps{H_F}_{V'}$.
\end{enumerate}
\end{Definition}
Similarly, as above, we have that $\Gamma(\ps{H}_F)=H_F$. 

We now define the main presheaf of our analysis:

\begin{Definition}
The presheaf $\ps{H/H_F}\in \Setnf$ is defined on
\begin{enumerate}
\item Objects: for each $V\in \cvnf $ corresponds the set \be
\ps{H/H_F}_V := H_V/H_{FV}= H/H_{FV}=\{[g]_V|\,g\sim g_1\;\, \text{ iff }\,\, hg_1=g \text{ for } h\in H_{FV}\}\,.
\ee  
\item Morphisms: given a morphism $i_{V'V}:V'\subseteq V$, the corresponding presheaf morphism is the map $\ps{H/H_F}(i_{V'V}): H/H_{FV}\rightarrow H/H_{FV'}$ defined as the bundle map of the bundle $H_{FV'}/H_{FV}\rightarrow H/H_{FV}\rightarrow H/H_{FV'}$. 

\end{enumerate} 
\end{Definition}
We can also define the following presheaf:
\begin{Definition}
The presheaf $\ps{H}/\ps{H_F}$ over $\cvnf$ is defined on 
\begin{itemize}
\item Objects: for each $V\in \cvnf$ we obtain $(\ps{H}/\ps{H_F})_V:=H/H_{FV}$, since the equivalence relation is computed context wise.
\item Morphisms: for each map $i:V^{'}\subseteq V$ we obtain the morphisms 
\ba
(\ps{H}/\ps{H_F})_V&\rightarrow& (\ps{H}/\ps{H_F})_{V^{'}}\\
H/H_{FV}&\rightarrow&H/H_{FV^{'}}
\ea
These are defined to be the projection maps $\pi_{V^{'}V}$ of the fibre bundles \be
H_{FV^{'}}/H_{FV}\rightarrow H/H_{FV}\rightarrow H/H_{FV^{'}} \ee
with fibre isomorphic to $H_{FV^{'}}/H_{FV}$.
\end{itemize}
\end{Definition}

We then obtain the analogue of Theorem 6.1 in \cite{group}
\begin{Theorem}[\cite{group}]
\be
\ps{H/H_F}\simeq \ps{H}/\ps{H_F}\,.
\ee
\end{Theorem}

Since $\cvnf$ is a poset and, as such, it is equipped with the lower Alexandroff topology (which we denote as $\cvfm$), we then have the ordinary result 
\be
\Setnf\simeq \Sh(\cvfm)\,,
\ee
where $\Sh(\cvfm)$ is the topos of sheaves over $\cvfm$. Thus, each presheaf in $\Setnf$ is a sheaf in $ \Sh(\cvfm)$. From now on, we will simply denote $\cvfm$ by $\cvnf$. Moreover, using the 1:1 correspondence between sheaves and etal\'e bundles, 
the presheaf $\ps{H/H_F}$ yields the associated etal\'e bundle\footnote{\begin{Definition}\vspace{-.1in}
Given a topological space $X$, a bundle $p_{E}:E\rightarrow X$ is said to be etal\'e iff $p_A$ is a local homeomorphism. This means that, for each $e\in E$ there exists an open set $V$ with $e\in V\subseteq E$, such that $pV$ is open in $X$ and $p_{|V}$ is a homeomorphism $V\rightarrow pV$.
\end{Definition}\vspace{-.1in}
If for example $X=\Rl^2$ then for each point of a fibre there will be an open disc isomorphic to an open disc in $\Rl^2$. It is not necessary that these discs have the same size. Such a collection of open discs on each fibre are glued together by the topology on $E$. 
Another example of etal\'e bundles are covering spaces. However, although all covering spaces are etal\'e, it is not the case that all etal\'e bundles are covering spaces.
} 
$p:\Lambda(\ps{H/H_{F}})\rightarrow \cvnf$ whose bundle space $\Lambda(\ps{H/H_{F}}) $ can be given a poset structure, as follows:
\begin{Definition}
Given two elements $[g]_{V'}\in H/H_{FV'} $, $[g]_V\in H/H_{FV}$ we define the partial ordering by
\be
[g]_{V'}\leq [g]_{V} \;\;\text{ iff }\;\; p([g]_{V'})\subseteq  p([g]_V)
\;\;\text{ and }\;\;   [g]_V\subseteq  [g]_{V'}\,.
\ee
\end{Definition}
In addition, since each element $[g]$ corresponds uniquely to a faithful\footnote{Here, by faithful, we mean that, for each $V$,  we will only consider automorphisms of the algebra $\mn\supseteq V$, which do not leave any element $A\in V$ unchanged. } automorphisms $\alpha^g_{\rho}(t)\in H$, it is also possible to define the ordering in terms of such automorphisms. We denote the set of faithful automorphisms by $\autf(V,\mn)$, from $V \to \mn$.
\begin{Definition}
Given two faithful automorphisms $\alpha_1\in\autf (V,\mn)$ and $\alpha_2\in \autf(V',\mn)$, we define the partial ordering by
\be
\alpha_2\leq \alpha_1 \;\;\text{ iff }\;\; p(\alpha_2)\subseteq  p(\alpha_1)\;\;\text{ and }\;\; \alpha_2=\alpha_1|_{V'}\,.
\ee
\end{Definition}
In the following, switching from $[g]$ to the representative $\alpha^g_{\rho}(t)$, we will denote $\alpha^g_{\rho}(t)$ as $l_g$.

\begin{Theorem}
The map $I:\Shvfm\rightarrow \Sh(\Lambda( \ps{H/H_F}))$ is a
functor defined as follows:
\begin{enumerate}
\item [(i)] Objects: $\big(I(\ps{A})\big)_{[g]_V}:=\ps{A}_{\,l_g(V)}=\Big((l_g)^*(\ps{A})\Big)(V)$. If $[g]_{V'}\leq [g]_{V}$, then
\be
(I\ps{A}(i_{[g]_{V'},[g]_V})):=\ps{A}_{\,l_g(V),l_g(V')}:\ps{A}_{\,l_g(V)}\rightarrow \ps{A}_{\,l_g(V')}\,,
\ee
where $V=p([g]_V)$ and $V'=p([g]_{V'})$.
\item [(ii)] Morphisms: given a morphism $f:\ps{A}\rightarrow\ps{B}$ in $\Shvfm$, define the associated morphism in 
$\Sh(\Lambda (\ps{H/H_F}))$ as
\ba
I(f)_{[g]_V}:I(\ps{A})_{[g]_V}&\rightarrow& I(\ps{B})_{[g]_V}\\
f_{[g]_V}:\ps{A}_{\,l_g(p([g]_V))}&\rightarrow
&\ps{B}_{\,l_g(p([g]_V))} 
\ea
\end{enumerate}
\end{Theorem}
We are now able to map all the sheaves in $\Shvfm$ to sheaves in $\Sh(\Lambda(\ps{H/H_F}))$. Applying the functor $p!:\Sh(\Lambda(\ps{H/H_F}))\rightarrow \Shvfm$, we finally obtain sheaves on our fixed category $\cvnf$ (for details the reader should refer 
to \cite{group}).

In this context, the object $\ps{H}$ is indeed a group object and, as shown in \cite{complexnumbers}, is the topos analogue of a one parameter group taking its values in $\ps{\Rl}$. Moreover, by defining the relevant objects in our formalism through the above method, i.e. in terms of the composite functor $p!\circ I$, the action of $\ps{H}$ does not induce twisted presheaves. 

Having defined the automorphism group $\ps{H}$, we 
can define the KMS condition in terms of such a group.
This is the purpose of the next section.

\subsection{Internal KMS condition}
\label{subsect:internalkms}

We work in the topos $\Shvfm$ for which the state space is defined as follows:
\begin{Definition}
The spectral presheaf $\breve{\us}:=p!\circ I(\us)$ is defined on
\begin{itemize}
\item[(i)] Objects: for each $V\in\cvnf$ we have
\be
\breve{\us}_V:=\coprod_{[g]_V\in H/H_{FV}}\us(l_gV)\simeq\coprod_{\alpha\in \autf(V, \mn)}\us(\alpha V) \ee which represents the disjoint
union of the Gel'fand spectrum of all algebras related to $V$, via
a faithful group transformation.
\item[(ii)] Morphisms: given a morphism $i_{V'V}:V'\rightarrow V$ ($V'\subseteq V)$, in $\cvnf$ the corresponding spectral presheaf morphism is
\ba
\breve{\us}(i_{V'V}):\breve{\us}_{V}&\rightarrow& \breve{\us}_{V'}\\
\coprod_{\alpha_1\in \autf(V, \mn)}\us_{\alpha_1(V)}&\rightarrow&\coprod_{\alpha_2\in \autf(V', \mn)}\us_{\alpha_2(V')}
 \ea
such that, given $\lambda\in\us_{\alpha_1(V)}$, we obtain $\breve{\us}(i_{V'V})(\lambda):=\us_{\alpha_1(V),\alpha_2(V')}\lambda=\lambda|_{\alpha_2(V')}$. Thus, $\breve{\us}(i_{V'V})$ is actually a co-product of morphism $\us_{\alpha_1(V),\alpha_2(V')}$, one for each 
$\alpha_{1} \in \autf(V, \mn)$.
\end{itemize}
\end{Definition}
Similarly, analogous to the previous $\ps{[0,1]}$, 
the presheaf of OR functions can be defined as follows:
\begin{Definition}
The presheaf $ \breve{\ps{[0,1]}}:=p!\circ I(\ps{[0,1]})$ is defined on 
\begin{enumerate}
\item[(i)] Objects: for each $V\in\cvnf$, we have
\be
\breve{\ps{[0,1]}}_{V}:=\coprod_{[g]\in H/H_{FV}}\ps{[0,1]}_{\,l_gV}\simeq\coprod_{\alpha\in \autf(V, \mn)}\ps{[0,1]}_{\,\alpha V}\,.
\ee
\item[(ii)] Morphisms: given $i_{V'V}:V'\subseteq V$, the corresponding presheaf morphism is 
\ba
\breve{\ps{[0,1]}}(i_{V'V}):\breve{\ps{[0,1]}}_{V}&\rightarrow& \breve{\ps{[0,1]}}_{V'}\\
\coprod_{\alpha_1\in \autf(V, \mn)}\ps{[0,1]}_{\,\alpha_1(V)}&\rightarrow&\coprod_{\alpha_2\in \autf(V', \mn)}\ps{[0,1]}_{\,\alpha_2(V')}\,.
 \ea
\end{enumerate} 
\end{Definition}
Applying the composite functor $p!\circ I$ to the truth object $\tob^{\rho, r}$, we obtain the presheaf $\breve{\tob}^{\rho, r}:=p!\circ I(\tob^{\rho, r})$.
\begin{Definition}
The truth object presheaf $\breve{\ps{\t}}^{\rho, r}$ is defined on
\begin{itemize}
\item[(i)] Objects: for each $V\in\cvnf$ we have
\be
\breve{\tob}^{\rho, r}_V:=\coprod_{[g]_V\in H/H_{FV}}\tob^{\rho, r}(l_gV)\simeq\coprod_{\alpha\in \autf(V, \mn)}\tob^{\rho, r}(\alpha V) \ee 
which represents the disjoint union of the truth object defined for all algebras related to $V$, via
a faithful group transformation.
\item[(ii)] Morphisms: given a morphism $i:V'\rightarrow V$ ($V'\subseteq V)$, in $\cvnf$ the corresponding morphism is
\ba
\breve{\tob}^{\rho, r}(i_{V'V}):\breve{\tob}^{\rho, r}_{V}&\rightarrow& \breve{\tob}^{\rho, r}_{V'}\\
\coprod_{\alpha_1\in \autf(V, \mn)}\tob^{\rho, r}_{\,\alpha_1(V)}&\rightarrow&\coprod_{\alpha_2\in \autf(V', \mn)}\tob^{\rho, r}_{\,\alpha_2(V')}
 \ea
such that, given $\ps{S}\in\tob^{\rho, r}_{\alpha_1(V)}$, we have $\breve{\tob}^{\rho, r}(i_{V'V})(\ps{S}):=\tob^{\rho, r}_{\alpha_1(V),\alpha_2(V')}\ps{S}=\ps{S}_{|\downarrow \alpha_2(V')}$.
 $\breve{\tob}^{\rho, r}(i_{V'V})$ is actually a co-product of morphism $\tob^{\rho, r}_{\alpha_1(V),\alpha_2(V')}$, one for each 
$\alpha_1 \in \autf(V, \mn)$.
\end{itemize}
\end{Definition}

In this setting, given a sate $\rho$, the associated measure is
\ba
\breve{\mu}^{\rho}:\breve{\us}&\rightarrow&\Gamma \breve{\ps{[0,1]}}\\
\breve{\ps{S}}&\mapsto&\breve{\mu}^{\rho}(\breve{\ps{S}})
\ea
such that, for each $V\in \cvnf$, we have
\be
\breve{\mu}^{\rho}(\breve{\ps{S}})(V):=\coprod_{[g]\in H/H_{FV}}\mu^{\rho}\ps{S}(l_gV)=\coprod_{[g]\in H/H_{FV}}\mu^{\rho}(l_g^*\ps{S})(V)\,.
\ee
If $\rho$ is a KMS state, then $\mu^{\rho}(l_g^*\ps{S})(V)=\mu^{\rho}(\ps{S})(V)$. It follows that $\mu^{\rho}(\ps{S}):\cvnf\rightarrow[0,1]$ is constant on all $H$ related $V$, i.e. such a measure is constant on the orbits $H/H_{FV}$ defined for each $V$.  It is then reasonable to first of all define a measure on $\Sub(I(\us))$, where $I:\Shvfm\rightarrow \Sh(\Lambda(\ps{H/H_F}))$ was introduced above. Such a measure would be
\ba
\bar{\mu}^{\rho}:\Sub(I(\us))&\rightarrow&\Gamma( I\ps{[0,1]})\crcr
\ps{S}&\mapsto&\bar{\mu}^{\rho}\ps{S}
\ea
where $\bar{\mu}^{\rho}\ps{S}:\Lambda(\ps{H/H_F})\rightarrow [0,1]$ is such that, $\forall [g]\in \Lambda(\ps{H/H_F})$ we have $\bar{\mu}^{\rho}\ps{S}[g]:=\mu^{\rho}\ps{S}(l_g(V))=
\rho(\hat{P}_{ \ps{S}(l_g(V))})$.

Since for each $V\in\cvnf$ corresponds the orbit $H/H_{FV}$, each global element $\bar{\mu}^{\rho}(\ps{S})$, when restricted to such an orbit, gives the constant local section. In other words, 
\be
\bar{\mu}^{\rho}\ps{S}_{|H/H_{FV}}:\Lambda(\ps{H/H_F})(V)\rightarrow [0,1]
\ee
is constant on all $[g]\in \Lambda(\ps{H/H_F})(V)$. 
It is possible to consider $\bar{\mu}_{|H/H_{FV}}(I(\ps{S}))$ as a constant global section if we introduced the presheaf $\overline{\ps{[0,1]_V}}\in \Sets^{\Lambda(\ps{H/H_{FV}})^{\op}}$. Here $\ps{H/H_{FV}}$ is the category whose elements are the equivalence classes $[g]$, as in $\Lambda(\ps{H/H_F})$, but whose morphisms are now given by group multiplication, i.e. $g:[h]\rightarrow [h^{'}]$ iff $h^{'}=gh$ and $h\notin H_{FV}$.

\begin{Definition}
The presheaf $\overline{\ps{[0,1]_V}}\in\Sets^{\Lambda(\ps{H/H_{FV}})^{\op}}$ has as 
\begin{enumerate}
\item Objects: for each $[h]$, the respective presheaf set is 
\be
\overline{\ps{[0,1]_V}}([h]):=I(\ps{[0,1]})[h]=\ps{[0,1]}l_h(V)=\{f:\downarrow l_h(V)\rightarrow [0,1]|f \text{ is OR}\}\,.
\ee
\item Morphisms: given a map $g:[h]\rightarrow [h^{'}]$ such that $h^{'}=gh$ then we get a corresponding presheaf map
\ba
\overline{\ps{[0,1]_V}}(g):\overline{\ps{[0,1]_V}}([h^{'}])&\rightarrow& \overline{\ps{[0,1]_V}}([h])\crcr
f&\mapsto&l_{g-1}(f)
\ea
where $f:\downarrow l_h(V)\rightarrow [0,1]$ and  $l_{g^-1}f:\downarrow l_{g^{-1}h}V\rightarrow [0,1]$.
\end{enumerate}
 \end{Definition} 
To show that this is indeed a presheaf, we need to show that, given two maps $g_j:[h^{''}]\rightarrow [h^{'}]$ and $g_i:[h^{'}]\rightarrow [h]$, then the following equality holds:
\be
\overline{\ps{[0,1]_V}}(g_j\circ g_i)=\overline{\ps{[0,1]_V}}(g_i)\circ\overline{\ps{[0,1]_V}}(g_j)\,.
\ee
 Considering first the left hand side we obtain
 \ba
 \overline{\ps{[0,1]_V}}(g_j\circ g_i):\overline{\ps{[0,1]_V}}[h]&\rightarrow&\overline{\ps{[0,1]_V}}[h^{''}]\crcr
 f&\mapsto&l_{(g_jg_i)^{-1}}(f)=l_{g_i^{-1}}l_{g_j^{-1}}f
 \ea
 On the other hand the right hand side is
 \ba
 \overline{\ps{[0,1]_V}}(g_i)\circ\overline{\ps{[0,1]_V}}(g_j):\overline{\ps{[0,1]_V}}[h]&\rightarrow&\overline{\ps{[0,1]_V}}[h^{'}]\rightarrow \overline{\ps{[0,1]_V}}[h^{'}]\crcr
 f&\mapsto&l_{g_j^{-1}}f\mapsto l_{g_i^{-1}}l_{g_j^{-1}}f
 \ea
 Thus, indeed $\overline{\ps{[0,1]_V}}$ is a well defined presheaf.

Obviously, for each $V_i\in\cvnf$, one gets a presheaf $\overline{\ps{[0,1]_{V_i}}}\in \Sets^{\Lambda(\ps{H/H_{FV_i}})^\op}$.  
Using these new presheaves, we can then define the first condition on the measure associated to a KMS state to be the following: 
\be
\breve{\mu}^{\rho}(\breve{\ps{S}})(V)=\Gamma^c(\overline{\ps{[0,1]_V}})\,,
\ee
where $\Gamma^c(\overline{\ps{[0,1]_V}})$ denotes the global element with constant value $c$, which, in this case is $\mu^{\rho}(\ps{S})(l_gV)$.
This condition on the measure translates to the following condition on the truth object representing a KMS:
\be\label{equ:truthkms}
\breve{\tob}^{\rho, r}_V=\{\breve{\ps{S}}\in \Sub(\breve{\us}_{|\dV})|\, \forall V'\subseteq V,\;\breve{\mu}(\breve{\ps{S}})(V')=\Gamma^c(\overline{\ps{[0,1]_{V'}}})\text{ for } c=\mu^{\rho}(\ps{S})(V) \geq r\}\,.
\ee
Thus, the truth object associated to the KMS state represents the collection of sub-objects for which the measure is constant on the orbits $H/H_{FV}$, for each $V\in \cvnf$.
A moment of thought reveals that this is nothing but the following:
\be
\breve{\tob}^{\rho, r}_V=[[\at^*\tob^{\rho, r}]]_{ \mu^{\rho}\,,V}\,.
\ee
The action of the group $\ps{H}$ on $ \breve{\tob}^{\rho, r}$ is defined for each $V\in\cvnf$
\ba
\ps{H}_V\times  \breve{\ps{\mathbb{T}}}^{\rho, r}_V&\rightarrow& \breve{\ps{\mathbb{T}}}^{\rho, r}_V\\
(\at, \ps{S})&\mapsto&\at^*(\ps{S})
\ea
Thus restricting only to the first KMS condition, we can define a truth object satisfying such a condition as follows:

\begin{Definition}
\label{def:condition1}
Given a truth object $\breve{\tob}^{\rho, r}$, it satisfies the first KMS condition iff the following diagram commutes:
\[\xymatrix{
\breve{\tob}^{\rho, r}\ar[dd]_{ \alpha(t)^*}\ar[ddrr]^{\mu^{\rho}}&&\\
&&\\
\breve{\tob}^{\rho, r}\ar[rr]^{\mu^{\rho}}&&\breve{\ps{\Gamma[0,1]}}\\
}\]
for all $\alpha(t)\in \Lambda(\ps{H})$.
\end{Definition}

We are still missing the second condition for a KMS state, namely for any two elements $\hat{A},\hat{B}\in \mn_{\alpha}$, then 
$\langle\rho, \hat{A}\alpha_{t+i\gamma}\hat{B}\rangle=\langle\rho, \at\hat{B}\hat{A}\rangle.$
Here $\alpha_{t+i\gamma}$ belongs to the set $E=\{\alpha_{t+i\gamma}|\; t\in\Rl\,, \gamma\in [0,\beta]\}$ which represents an extension of the group $H$ to a strip on the complex plane. Since we are also considering $\gamma=0$, it follows that $H$ is contained in the above set. Is it possible to internalize $E$ as we did for $H$? The answer is yes and it is done through the construction of the following trivial presheaf $\ps{E}$ for which, at each context, we simply assign the set 
$E$ itself 
 and the maps are identity maps. Note that such presheaf is not a group, nevertheless we get $\ps{H}\subset \ps{E}$. Recalling the definition of the presheaf $\ps{H}_F$, it is possible to define its action on $\ps{E}$ as follows:
\be
\ps{H_{F}}\times\ps{E}\rightarrow\ps{E}\,,
\ee
such that, for each $V\in\cvnf$, we obtain
\ba
\ps{H_F}_{V}\times\ps{E}_V&\rightarrow&\ps{E}_V\\
(\alpha(t_1), \alpha(t_2+i\gamma))&\mapsto&\alpha(t_1+t_2, i\gamma)
\ea
Given this action, it is meaningful to define the quotient presheaf $\ps{E/H_F}$. Thus, for each $V\in \cvnf$ we obtains the set $E/H_{FV}$ where two elements $\alpha(t_1+i\gamma)\simeq \alpha(t_2+i\gamma)$ iff there exists an element $\alpha(t_3)$ such that $t_2=t_1+t_3$. Clearly, $\ps{H/H_{F}}\subset \ps{E/H_F}$. We can then extend what has been done with respect to the presheaf $\ps{H/H_{F}}$ to the presheaf $\ps{E/H_F}$. 

First, we define the new truth object for each $V\in\cvnf$ as
\be
\tilde{\tob}^{\rho,r}_{V}:=\coprod_{[g]\in E/H_{FV}}\tob^{\rho, r}(l_gV)\,.
\ee 
Here $\tilde{\tob}^{\rho,r}$ is defined as\footnote{Note
that we use the same method that was used to define $\breve{\ps{\t}}^{\rho, r}$, but now we replace $I$ by $\mathcal{I}$ and $p!$ by $p_E!$. }
$\tilde{\tob}^{\rho,r}:=p_E!\circ \mathcal{I} (\tob^{\rho,r})$ where $\mathcal{I}:\Shvfm\rightarrow \Sh(\Lambda(\ps{E/H_F}))$ and $p_E:\Lambda(\ps{E/H_F})\rightarrow \cvnf$.
Since $\ps{H}\subset \ps{E}$, the above truth object contains all the elements of $\breve{\tob}^{\rho}_{(V,r)}$. Hence, the first KMS condition keeps its form.
On the other hand, the second KMS condition seems a little more involved since it concerns measures on the conjunction of two sub-objects of the state space. To implement the second KMS condition,  the following natural transformation is needed:
\be\label{equ:f}
F:\ps{E}\times\tilde{\tob}^{\rho}\rightarrow\tilde{\tob}^{\rho}\,,
\ee
However, we need to pay attention to domain issues. In particular we know that  $F_{\ps{T},\ps{S}}(t+i\beta)=[\mu^{\rho}(\ps{T}\wedge\alpha(t+i\beta)^*\ps{S}]\mn_{\alpha}$ only for $\mn_{\alpha}$ being a $\sigma$-weakly $\alpha$-invariant sub-algebra of $\mn$. Therefore we can construct a category $\mv(\mn_{\alpha})$ of abelian sub-algebras of  $\mn_{\alpha}$ ordered by inclusion. Clearly $\mv(\mn_{\alpha})$ is a full subcategory of $\mv(\mn)$, thus we can define the continuous identity map $i:\mv(\mn_{\alpha})\rightarrow\mv(\mn)$. This gives rise to the corresponding geometric morphisms between $i:Sh(\mv(\mn_{\alpha}))\rightarrow Sh(\mv(\mn))$ whose associated inverse image is  $i^*:Sh(\mv(\mn))\rightarrow Sh(\mv(\mn_{\alpha}))$ such that $(i^*(\ps{A}))_{V_{\alpha}}:=\ps{A}_{i(V_{\alpha})}=\ps{A}_{V_{\alpha}}$.

We then utilise these inverse image morphisms to correctly define the natural transformation in \eqref{equ:f}, which we still call $F$, as
\ba
F:i^*(\ps{E}\times \tilde{\tob}^{\rho})&\rightarrow&i^*(\tilde{\tob}^{\rho})\\
F:i^*(\ps{E})\times i^*(\tilde{\tob}^{\rho})&\rightarrow&i^*(\tilde{\tob}^{\rho})
\ea
where the last equation holds since the left adjoint $f^*$ preserves finite limits. Given a context $V_{\alpha}\in\mv(\mn_{\alpha})$ we then have 
\ba
F_V:(i^*(\ps{E}))_{V_{\alpha}}\times( i^*(\tilde{\tob}^{\rho})_{V_{\alpha}}&\rightarrow&(i^*(\tilde{\tob}^{\rho})_{V_{\alpha}}\crcr
F_V:\ps{E}_{V_{\alpha}}\times \tilde{\tob}^{\rho}_{V_{\alpha}}&\rightarrow&\tilde{\tob}^{\rho}_{V_{\alpha}}\crcr
(\alpha(t+i\beta), \ps{S})&\mapsto&F((\alpha(t+i\beta), \ps{S})):=\alpha^*(t+i\beta) \ps{S}\,.
\ea
To show that this is a well defined natural transformation, we prove that, given a pair of contexts $V'$ and $V$, such that $i_{V'V}:V'\subseteq V$,  the following diagram commutes\footnote{Note that, since $(i^*(\ps{A}))_{V_{\alpha}}=\ps{A}_{i(V_{\alpha})}=\ps{A}_{V_{\alpha}}$, form now on we will simply write $\ps{A}_{V_{\alpha}}$ and the action of the $i^*$ functor will be assumed.}:
\[\xymatrix{
\ps{E}_{V_{\alpha}}\times\tilde{\tob}^{\rho}_{V_{\alpha}}\ar[rr]^{F_{V_{\alpha}}}\ar[dd]_{\langle\ps{E}(i_{V'_{\alpha}V_{\alpha}}), \tilde{\tob}^{\rho}(i_{V'_{\alpha}V_{\alpha}})\rangle}&&\tilde{\tob}^{\rho}_{V_{\alpha}}\ar[dd]^{ \tilde{\tob}^{\rho}(i_{V'_{\alpha}V_{\alpha}})}\\
&&\\
\ps{E}_{V'_{\alpha}}\times\tilde{\tob}^{\rho}_{V'_{\alpha}}\ar[rr]^{F_{V'_{\alpha}}}&&\tilde{\tob}^{\rho}_{V'_{\alpha}}\\
}\]
Chasing the diagram round one way, we obtain:
\be
\Big[ \tilde{\tob}^{\rho}(i_{V'_{\alpha}V_{\alpha}})\circ F_V\Big]\big(\alpha(t), \ps{S}\big)=\tilde{\tob}^{\rho}(i_{V'_{\alpha}V_{\alpha}})\big(\alpha(t)^*\ps{S}\big)=(\alpha(t)^*\ps{S})_{|\dV'_{\alpha}}=\ps{S}_{|\alpha(t)V'_{\alpha}}
\ee
while going round the other way it yields:
\be
F_{V'}\circ \langle\ps{E}(i_{V'_{\alpha}V_{\alpha}}), \tilde{\tob}^{\rho}(i_{V'_{\alpha}V_{\alpha}})\rangle\big(\alpha(t), \ps{S}\big)=F_{V'_{\alpha}}\big(\alpha(t)_{|V'_{\alpha}}, \ps{S}_{|\dV'_{\alpha}}\big)=\alpha(t)^*\ps{S}_{|\dV'_{\alpha}}=\ps{S}_{|\alpha(t)V_{\alpha}'}\,.
\ee
Hence, $F$ is a natural transformation.

We can now write the second condition for the KMS state as follows:

\begin{Definition}
The truth object $\tilde{\tob}^{\rho}$ represents the topos analogue of the KMS state iff for all $V_{\alpha}\in\mv(\mn)$, the following diagram commutes
\[\xymatrix{
\tilde{\tob}^{\rho, r}_{V_{\alpha}}\times \tilde{\tob}^{\rho, r}_{V_{\alpha}}\ar[rr]^{\langle\at^*,id\rangle}
\ar[dd]_{\langle\alpha_{t+i\gamma}^*, id \rangle}&&\tilde{\tob}^{\rho, r}_{V_{\alpha}}\times \tilde{\tob}^{\rho, r}_{V_{\alpha}}\ar[rr]^{f}&&\tilde{\tob}^{\rho, r}_{V_{\alpha}}\times \tilde{\tob}^{\rho, r}_{V_{\alpha}}\ar[dd]^{\wedge}\\
&&&&\\
\tilde{\tob}^{\rho, r}_{V_{\alpha}}\times \tilde{\tob}^{\rho, r}_{V_{\alpha}}\ar[rrrr]^{\mu^{\rho}}\ar[dd]^{\wedge}&&&&\tilde{\tob}^{\rho, r}_{V_{\alpha}}\ar[dd]^{\mu^{\rho}}\\
&&&&\\
\tilde{\tob}^{\rho, r}_{V_{\alpha}}\ar[rrrr]^{\mu^{\rho}}&&&& \tilde{\ps{\Gamma[0,1]}}_{V_{\alpha}}\\
}\]
for all $\alpha_{t+i\gamma}\in \Lambda(\ps{E})$. Here $f$ is a switching function,  i.e. $f:\tilde{\tob}^{\rho, r}_{V_{\alpha}}\times \tilde{\tob}^{\rho, r}_{V_{\alpha}}\rightarrow \tilde{\tob}^{\rho, r}_{V_{\alpha}}\times \tilde{\tob}^{\rho, r}_{V_{\alpha}}$; $f(\ps{S}, \ps{T}):=(\ps{T}, \ps{S})$.

\end{Definition}

Thus given a pair $(\ps{S}, \ps{T})\in\tilde{\tob}^{\rho, r}\times \tilde{\tob}^{\rho, r}$, going one way round the diagram, we obtain
\be
\mu^{\rho}\circ\wedge\circ f\circ\langle\at ^*,id\rangle(\ps{S}, \ps{T})_{V_{\alpha}}=\mu^{\rho}\circ\wedge\circ f((\at ^*\ps{S})_{V_{\alpha}}, \ps{T}_{V_{\alpha}})=\mu^{\rho}\circ\wedge(\ps{T}_{V_{\alpha}},\at ^*\ps{S}_{V_{\alpha}})=\mu^{\rho}(\ps{T}\wedge\at ^*\ps{S})_{V_{\alpha}}\,.
\ee
Going the other way instead, it can be proved that
\be
\mu^{\rho}\circ\wedge\circ\langle\alpha_{t+i\gamma}^*, id \rangle(\ps{S}, \ps{T})_{V_{\alpha}}=\mu^{\rho}\circ
\wedge((\alpha_{t+i\gamma}^*\ps{S})(V), \ps{T}_{V_{\alpha}})=
\mu^{\rho}(\alpha_{t+i\gamma}^*\ps{S}\wedge \ps{T})_{V_{\alpha}}\,.
\ee
This is precisely the second KMS condition.

It is straightforward to see that the above definition incorporates Definition \ref{def:condition1} by setting $\gamma=0$.
So far, we have considered the one parameter family of truth objects $\tilde{\tob}^{\rho, r}_V$ one for each $r\in(0,1)_L$. As stated in the Appendix \ref{app:truth}, these can be combined  to define an object in the topos $\Sh(\cvnf\times (0,1)_L)$ as follows:
\begin{Definition}
The truth object $\tilde{\tob}^{\rho}\in \Sh(\cvnf\times (0,1)_L)$  is defined on:
\begin{enumerate}
\item[(i)] Objects: for each pair $(V, r)\in \cvnf\times (0,1)_L$, we 
have
\be
\tilde{\tob}^{\rho}_{(V,r)}:=\coprod_{[g]\in E/H_{FV}}\tob^{\rho, r}_{l_g(V)} \simeq \coprod_{\at \in \autf(V, \mn)}\tob^{\rho, r}_{\at (V)}\ee 
where each 
\be
\tob^{\rho,r}_{\,\at (V)}=\{\ps{S}\in \Sub(\us_{\downarrow\at (V)})|\,\forall 
V' \subseteq\at (V),\;
(\mu^{\rho}\ps{S})(V')\geq r\} \,.
\ee
\item[(ii)] Morphisms: given a map $i:(V', r')\leq (V, r)$ (iff $V'\subseteq V$ and $r'\leq r$), then the associated morphism is
\ba
 \tilde{\tob}^{\rho}(i):\coprod_{[g]\in H/H_{FV}}\tob^{\rho,r}_{\,l_g(V)}&\rightarrow& \coprod_{[h]\in H/H_{FV'}}\tob^{\rho,r}_{\, l_h(V')}\crcr
 \coprod_{\alpha^1(t)\in \autf(V, \mn)}\tob^{\rho, r}_{\,\alpha^1(t)(V)}&\rightarrow&\coprod_{\alpha^2(t)\in \autf (V', \mn)}\tob^{\rho, r}_{\, \alpha^2(t)(V')}
  \ea 
 such that, given a sub-object $\ps{S}\in\tob^{\rho,r}_{\,\alpha^1(t)(V)} $, we get
\be
\breve{\ps{\mathbb{T}}}^{\rho}(i)\ps{S}:=\ps{\mathbb{T}}^{\rho,r}(i_{\alpha^1(t)(V),\alpha^2(t)(V')})\ps{S}=\ps{S}_{|\alpha^2(t)(V')}
\ee 
where 
$\alpha^2(t)\leq \alpha^1(t)$ thus 
$p(\alpha^2(t))\subseteq p(\alpha^1(t))$ and $\alpha^2(t)(V')=(\alpha^1(t))|_{V'}(V')$. 
Obviously, now the condition on the restricted sub-object is
$(\mu^{\rho}\ps{S})(V'')\geq r'$ where $V''\subseteq
\alpha^2(t)(V')$. However, such a condition is trivially satisfied since
$r'\leq r$.
\end{enumerate}
\end{Definition}
Given the internal definition of the topos KMS state, then it is possible to derive conditions on expectation values and truth values in the same was as was done for the external perspective.  Moreover the derivation of the standard KMS state from the topos KMS state is done similarly as for the internal situation.

\section{Connection with Tomita-Takesaki modular theory}
\label{sect:tttheo}

We briefly analyze the connection between the topos formulation of KMS state and the Tomita-Takesaki formulation. The topos role of the operators $\Delta$ and $J$ has to be identified. 
We will focus only on the external formulation, but it is clear
that there should exist a corresponding internal formulation of the
following.

In Tomita-Takesaki theory, $\Delta$ is a self-adjoint operator called modular operator and $J$ is an antiunitary operator.  As such, in our topos $\Setn$, $\Delta$  should be represented  as a map between the state space $\us$ and the real number quantity value object $\ps{\Rl}^{\leftrightarrow}$, i.e.
$
\breve{\delta{\Delta}}:\us\rightarrow \ps{\Rl}^{\leftrightarrow}
$, whereas $J$ should be  represented as 
$
\breve{\delta}(J):\us\rightarrow\ps{\Cl}^{\leftrightarrow}.
$ 
Nevertheless, such a description, following from \cite{andreas5}, 
does not seem to be useful in the present analysis. We proceed
differently in the following.

\begin{Definition}
The presheaf $\ps{\Delta}\in \Setn$  is defined on:
\begin{enumerate}
\item [(i)] Objects: for each $V\in\cvn$, we have the set $\ps{\Delta}_V:=\{\Delta^{it}|\, t\in\Rl\}$ where $\Delta^{it}$ is a unitary operator and 
$\ps{\Delta}_V$ forms a strongly continuous unitary group.
\item [(ii)] Morphisms: for a given morphism $i_{V'V}:V'\subseteq V$, the corresponding presheaf morphism is simply the identity map.
\end{enumerate}
\end{Definition}
 Such a presheaf represents a group in $\Setn$ and is clearly the more abstract definition of the group $\ps{H}$. Thus, when defining the topos analogue of the KMS state, if we want to have an explicit connection with Tomita-Takesaki formulation, $\ps{H}$ should be replaced by $\ps{\Delta}$.

On the other hand, to describe $J$ in the topos formulation, its role 
has to be understood better. In particular, let us consider a von Neumann algebra $\mn$, such that it contains two sub-algebras which we call $\mn_l$ and $\mn_r$. These sub-algebras are such that i) $\mn_l\cap\mn_r=id_{\mn}$, ii) $\mn_l=(\mn_r)^{''}$ ($\mn_r=(\mn_l)^{''}$) and iii) 
\be
\label{equ:level}
J\mn_lJ=\mn_r\,.
\ee
It can be underlined that this description is not physically insignificant. Indeed, by applying the modular theory to a particle in a plane submitted to a magnetic field  perpendicular to that plane, 
 $\mn_{l,r}$ algebras correspond to algebra of Hilbert-Schmidt 
operators built out of the two sectors of the Hilbert
space of states (Landau levels) of the system when the magnetic field point towards the ``up'' or ``down'' direction \cite{ali}. 
In this setting the transformation  \eqref{equ:level} represents a change in the magnetic direction. 

Starting from \eqref{equ:level},  at the level of categories, 
a map can be defined as follows:
\ba\label{ali:jmap}
P:\cV(\mn_l)&\rightarrow&\cV(\mn_r)\\
V_l&\mapsto&JV_lJ
\ea
where $\cV(\mn_l)$ and $\cV(\mn_r)$ are the categories of von Neumann abelian sub-algebras of the algebras $\mn_l$ and $\mn_r$ respectively. Both such category are ordered by inclusion.
We intend to show that the map $P$ is continuous, to this end we need to define a topology on the categories $\cV(\mn_l)$ and $\cV(\mn_r)$. Since they are both posets under algebra inclusion, they are automatically equipped with the Alexandroff topology whose basis are the lower sets $\dV:=\{V'|V'\subseteq V\}$.
In order to show continuity, we define an inverse which is trivially
\ba
P^{-1}:\cV(\mn_r)&\rightarrow&\cV(\mn_l)\\
V_r&\mapsto&J^{-1}V_rJ^{-1}
\ea 
\begin{Theorem}
The map $P$ is continuous with respect to the Alexandroff topology.
\end{Theorem}
\begin{proof}
Given an open set (it suffices to show this for the basis sets)  $\dV_r$, the pullback is $P^{-1}(\dV_r)$ and is defined as 
\be
P^{-1}(\dV_r):=\{P^{-1}V'_r|\,V'_r\subseteq V_r\}\,.
\ee
Since $V'_r\subseteq V_r$ iff $P^{-1}V'_r\subseteq P^{-1}V_r$, the above equation can be written as
\be
P^{-1}(\dV_r):=\{P^{-1}V'_r|\,P^{-1}V'_r\subseteq P^{-1}V_r\}=\downarrow P^{-1}V_r=\dV_l\,.
\ee
\end{proof}
\begin{Lemma}
Given the map $P:\cV(\mn_l)\rightarrow \cV(\mn_r)$ between the two posets, then $P$ is order preserving iff it is continuous with respect to the Alexandroff topology, i.e. for each lower subset $\downarrow L\in\cV(\mn_r)$ then $P^{-1}(\downarrow L)$ is a lower subset in $\cV(\mn_l)$.
\end{Lemma}
\begin{proof}
Let us assume that $P$ is order preserving and $\dV_r$ is a lower set in $\cV(\mn_r)$. Moreover consider $V_l^1\in P^{-1}(\dV_r)$ which is equivalent to $P(V_l^1)=V_r^1\in \dV_r$. We assume that $V^2_l\in\cV(\mn_l)$ such that $V^2_l\leq V^1_l$. 
Since $P$ is order preserving by assumption, then $P(V^2_l)\leq P(V_l^1)=V_r^1\in \dV_r$. Since $\dV_r$ is a lower set it means that  $P(V^2_l)\in \dV_r$, i.e. $V^2_l\in P^{-1}(\dV_r)$. It follows that  $P^{-1}(\dV_r)$ is a lower set, i.e. $P$ is continuous.

Conversely, let us assume that $P$ is continuous, then, given any lower set $\dV_r\in\cV(\mn_r)$, $P^{-1}\dV_r$ is a lower set in $\cV(\mn_l)$ and that there is a pair $V_l^1, V_l^2\in\cV(\mn_l)$ such that $V_l^1\leq V_l^2$. Since $\downarrow P(V_l^2)$ is a lower set in $\cV(\mn_r)$ then $P^{-1}\downarrow P(V_l^2)$ is a lower set in $\cV(\mn_l)$. However $P(V_l^2)\in \downarrow P(V_l^2)$, therefore $V_l^2\in P^{-1}\downarrow P(V_l^2)$. But  $V_l^1\leq V_l^2$ implies that $V_l^1\in P^{-1}\downarrow P(V_l^2)$ ($P(V_l^1)\in\downarrow P(V_l^2)$) which means that $P(V_l^1)\leq P(V_l^2)$. Thus $P$ is order preserving.
\end{proof}

Since $P$ is continuous, we can lift it to the level of sheaves and define a geometric morphism whose direct image is \ba
P_*:\Sh(\cV(\mn_l))&\rightarrow&\Sh(\cV(\mn_r))\\
\ps{S}&\mapsto&P_*(\ps{S}):=\ps{S}\circ P^{-1}
\ea
such that, for each context $V_r\in\cV(\mn_r)$, we have $(P_*(\ps{S}))(V_r)=\ps{S}\circ P^{-1}(V_r)=\ps{S}(J^{-1}V_rJ^{-1})$.
On the other hand the inverse image is
\ba
P^*:\Sh(\cV(\mn_r))&\rightarrow&\Sh(\cV(\mn_l)\\
\ps{A}&\mapsto&P^*(\ps{A}):=\ps{A}\circ P
\ea
such that, for each context $V_l\in\cV(\mn_l)$, we have $(P^*(\ps{A}))(V_l)=\ps{A}\circ P(V_l)=\ps{S}(JV_lJ)$.
The above maps are adjoint of each another in particular $P^*$ is the left adjoint of $P_*$ which is denoted as $P^*\dashv P_*$. The properties of adjoints are the following: 
To each morphism $f:X_r\rightarrow P_*A_l$ in $\Sh(\cV(\mn_r))$, there is associated the unique morphism $h:P^*X^r\rightarrow A_l$ in $\Sh(\cV(\mn_l)$. This means is that the following 
\be
\theta:\Hom_{\Sh(\cV(\mn_r))}(X_r, P_*A_l)\rightarrow \Hom_{\Sh(\cV(\mn_l))}(P^*X_r, A_l)
\ee
is an isomorphism for any element $X_r\in \Sh(\cV(\mn_r))$ and $A_l\in \Sh(\cV(\mn_l))$. 
We conjecture at this point that the change in the magnetic direction can be described in terms of a geometric morphism. The consequences of such a property will be 
analyzed in a forthcoming work. 

\section*{Acknowledgements}
C.F. would like to thank her
mother Elena Romani, her father Luciano Flori and her grandmother Pierina Flori for constant support and encouragement.
 
Research at Perimeter Institute is supported by the Government of Canada through Industry
Canada and by the Province of Ontario through the Ministry of Research and Innovation.

\bigskip
\bigskip

\section*{Appendix}
\label{app}

\appendix

\renewcommand{\theequation}{\Alph{section}.\arabic{equation}}
\setcounter{equation}{0}

This appendix gathers basic notions in category and topos theories,
precisions on notations and topos applications to any quantum phase space. 
Well-known results used in the text have been, for the sake
of clarity, rederived.

\section{Short introduction to category theory}
\label{app:category}

\subsection{Categories, arrows, posets and all that}

{\it ``Category Theory allows you to work on structures without the need first to pulverize them into set theoretic dust''} (Corfiel).

The above quotation explains, in a rather pictorial way, what {\it category theory} and, in particular, {\it topos theory} are really about.
In fact, category theory (and topos theory) allows to abstract from the specification of points (elements of a set) and 
functions between these points to a universe of discourse in which the basic elements are arrows, and any property is given 
in terms of compositions of arrows.

The reason for the above characterization is that the underlying philosophy behind category theory (and topos theory) is that of describing mathematical objects from an external point of view, i.e. in terms of relations.

This is in radical contrast to {\it set theory}, whose approach is essentially internal in nature. By this we mean that the basic/primitive notions of what sets are and the belonging relations between sets, are defined in terms of the elements which belong to the sets in question, i.e. an internal perspective.

In order to be able to implement the notion of external definition, we first need to define two important notions: i) the notion of a map or arrow, which is simply an abstract characterization\footnote{By abstract characterization here we mean a notion that does not depend on the sets or objects between which the arrow is defined.} of the notion of a function between sets; ii) the notion of an {\it equation} in categorical language. We will first start with the notion of a map. For detailed analysis see \cite{topos3,topos8, bell, mythesis}.

Given two general objects $A$ and $B$ (not necessarily sets) an arrow $f$ is said to have domain $A$ and codomain $B$ if it goes from $A$ to $B$, i.e. $f:A\rightarrow B$. It is convention to denote $A={\rm dom}(f)$ and $B={\rm cod}(f)$.
We often draw such an arrow as follows:
\be
A\xrightarrow{f}B
\ee

Given two arrows $f:A\rightarrow B$ and $g:B\rightarrow C$, such that ${\rm cod}(f)={\rm dom}(g)$ then we can {\it compose} the two arrows obtaining $g\circ f:A\rightarrow C$. The property of {\it composition} is drawn as follows:
\be
A\xrightarrow{f}B\xrightarrow {g}C
\ee
For each object (or set $A$), there always exists and {\it identity arrow} \be
id_A:A\rightarrow A\,,\quad \text{such that}
\quad id_A(a)=a\,, \quad \forall a\in A\,.
\ee
The collection of arrows between various objects satisfies two laws:
\begin{enumerate}
\item [i)] {\it Associativity}: given three arrows $f:A\rightarrow B$, $g:B\rightarrow C$ and $h:C\rightarrow D$ with appropriate domain and codomain relations,  we then have
\be
h\circ (g\circ f)=(h\circ g)\circ f\,.
\ee
\item [ii)] {\it Unit law}: given $f:A\rightarrow B$, $id_A:A\rightarrow A$ and $id_B:B\rightarrow B$ the following holds
\be
f\circ id_A=f=id_B\circ f\,.
\ee
\end{enumerate}

We can now give the axiomatic definition of a category:
\begin{Definition}
A (small\footnote{A category $\mathcal{C}$ is called small if $Ob(\mathcal{C})$ is a Set. }) category $\mathcal{C}$ consists of the following elements:
\begin{itemize}
\item[1.] A collection $Ob(\c)$ of $\c$-objects.
\item[2.] For any two objects $a, b\in Ob(\c)$, a set $Mor_{\c}(a, b)$ of $\c$-arrows (or $\c$-morphisms) from $a$ to $b$.
\item[3.] Given any three objects $a, b, c\in\c$, a map which represents composition operation.
\ba
\circ:Mor_{\c}(b,c)\times Mor_{\c}(a,b)&\rightarrow& Mor_{\c}(a,c)\\
(f,g)&\mapsto&f\circ g
\ea
Composition is {\it associative}, i.e. for $f\in Mor_{\c}(b,c)$, $g\in Mor_{\c}(a,b)$ and $h\in Mor_{\c}(c,d)$ we have
\be
h\circ(f\circ g)=(h\circ f)\circ g
\ee
which in diagrammatic form is the statement that the following diagram commutes
\[\xymatrix{
\ar[dd]_{h\circ(f\circ g)}&&a\ar[rr]^{g}\ar[ddrr]_{\;\;\;\;h\circ f\;\;\;}\ar[dd]_{(h\circ f)\circ g}&&b\ar[dd]^{f}\ar[ddll]^{\;\;\;f\circ g}\\
&&&&\\
&&d&&c\ar[ll]^{h}\\
}\]
\item[4.] For each object $b\in\c$, an identity morphism $id_{b}\in Mor_{\c}(b,b)$, such that the following {\it Identity law} holds: for all $g\in Mor_{\c}(a,b)$ and $f\in Mor_{\c}(b, c)$ then $f=f\circ id_b$ and $g=id_b\circ g$. In diagrammatic form, this is represented by the fact that the diagram
\[\xymatrix{
a\ar[rr]^{g}\ar[ddrr]^{g}&&b\ar[dd]^{id_b}\ar[ddrr]^{f}&&\\
&&&&\\
&&b\ar[rr]^{f}&&c\\
}\]
commutes.
\end{itemize}
\end{Definition}
So, a category is essentially a collection of diagrams for which certain equations or commutative relations hold.

An example of a category which is used throughout the paper is a {\it poset}. 
A poset is a set in which the elements are related by a partial order, i.e. not all elements are related to each other. The definition of a poset is as follows:
\begin{Definition}
Given a set $P$ we call this a poset iff a partial order $\leq$ is defined on it.  A partial order is a binary relation $\leq$ on a set $P$, which has the following properties:
\begin{itemize}
\item Reflexivity: $a\leq a$ for all $a\in P$.
\item Antysimmetry: if $a\leq b$ and $b\leq a$, then $a=b$.
\item Transitivity: If $a\leq b$ and $b\leq c$, then $a\leq b$.
\end{itemize}

\end{Definition}
An example of a poset is any set with an inclusion relation defined on it. Another example is $\Rl$ with the usual ordering defined on it.  

A poset forms a category whose objects are the elements of the poset and, given any two elements $p, q$, there exists a map $p\rightarrow q$ iff $p\leq q$ in the poset ordering.
We will be using such a (poset) category quite often when defining a topos description of quantum theory. Thus, it is worth pointing out the following:
\begin{Definition}
Given two partial ordered sets $P$ and $Q$, a map/arrow $f:P\rightarrow Q$ is a partial order homomorphism (otherwise called monotone functions or order preserving functions) if
\be
\forall x,y\in P\,, \quad \;\;x\leq y\,\,\Longleftrightarrow \,\,f(x)\leq f(y)\,.
\ee
\end{Definition}
Homomorphisms are closed under composition. A trivial example of partial order homomorphisms is given by the identity
maps.

We now define some basic notions in category theory. 

\begin{Definition} 
A \textbf{terminal object} in a category $\c$ is a $\c$-object 1 such that, 
given any other $\c$-object A, there exists one and only one $\c$-arrow from A 
to 1. 
\end{Definition}
For example, in $\Sets$, a terminal object is a singleton $\{*\}$. Since given any other element $A\in \Sets$, there exists one and only one arrow $A\rightarrow\{*\}$. On the other hand, in \textbf{Pos}, the poset $(\{*\}, \{(*, *)\})$ is the terminal object.

\subsection{Functors and natural transformations}

It is also possible to define maps between categories, such maps are called {\it functors}. Generally speaking a functor is a transformation from one category $\c$ to another category $\d$, such that the categorical structure of the domain $\c$ is preserved, i.e. gets mapped onto the structure of the codomain category
$\d$.\\
There are two types of functors {\it covariant functors} and {\it contravariant functor}.

\begin{Definition}:
A \textbf{covariant functor} from a category $\c$ to a category $\d$ is a map 
$F:\c\rightarrow\d$ that assigns to each $\c$-object $a$, a 
$\d$-object F(a) and to each $\c$-arrow $f:a\rightarrow b$ a $\d$-arrow
$F(f):F(a)\rightarrow F(b)$, such that the following are satisfied:
\begin{enumerate}
\item $F(id_a)=id_{F(a)}$
\item $F(f \circ g)=F(f) \circ F(g)$ for any $g:c\rightarrow a$ 
\end{enumerate}
\end{Definition}

It is clear from the above that a covariant functor is a transformation that preserves both:
\begin{itemize}
\item The domain's and the codomain's identities.
\item The composites of functions, i.e. it preserves the direction of the arrows.
\end{itemize}
A pictorial description for a covariant functor is given as follows:
 
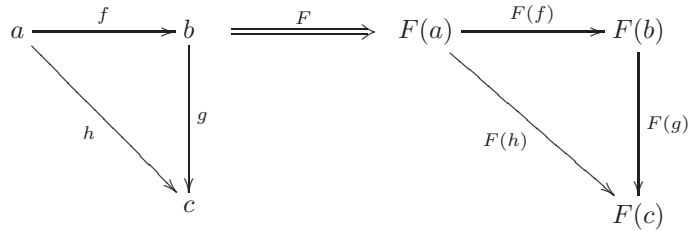
\[\xymatrix{
a\ar[rr]^f\ar[rrdd]_h&&b\ar[dd]^g\\
&&\\
&&c\\
}
\xymatrix{\ar@{=>}[rr]^F&&}
\xymatrix{
F(a)\ar[rr]^{F(f)}\ar[rrdd]_{F(h)}&&F(b)\ar[dd]^{F(g)}\\
&&\\
&&F(c)\\
}\]
On the other hand, we have
\begin{Definition}

A \textbf{contravariant functor} from a category $\c$ to a category $\d$ is a 
map $X:\c\rightarrow\d$ that assigns to each $\c$-object a a 
$\d$-object X(a) and to each $\c$-arrow $f:a\rightarrow b$ a $\d$-arrow
$X(f):X(b)\rightarrow X(a)$, such that the following conditions are satisfied:
\begin{enumerate}
\item $X(id_a)=id_{X(a)}$
\item $X(f \circ g)=X(g) \circ X(f)$ for any $g:c\rightarrow a$ 
\end{enumerate}
\end{Definition}
A diagrammatic representation of a contravariant functor is the following:

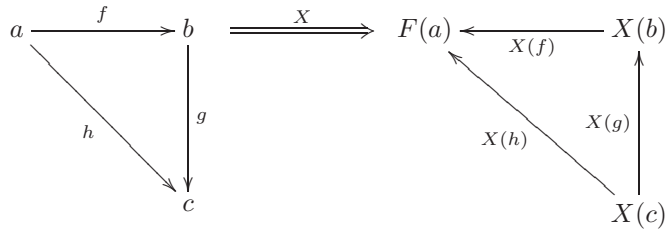
\[\xymatrix{
a\ar@{>}[rr]^f\ar[rrdd]_h&&b\ar[dd]^g\\
&&\\
&&c\\
}
\xymatrix{\ar@{=>}[rr]^X&&}
\xymatrix{
F(a)&&X(b)\ar[ll]^{X(f)}\\
&&\\
&&X(c)\ar[lluu]^{X(h)}\ar[uu]^{X(g)}\\
}\]

\vspace{.2in}
Thus, a contravariant functor in mapping arrows from one category 
to the next reverses
the directions of the arrows, by mapping domains to codomains and vice versa. A contravariant functor is also called a {\it presheaf}.\\
So far we have defined categories and maps between them called {\it functors}. We will now abstract a step more and define maps between functors. These are called {\it natural transformations}.
 \begin{Definition} 
A \textbf{natural transformation} from $Y:\c\rightarrow \d$ to 
$X:\c\rightarrow \d$ 
is an  assignment of an arrow 
$N:Y\rightarrow X$ that associates to each object A in $\c$ an arrow 
$N_A:Y(A)\rightarrow X(A)$ in $\d$ such that, for any 
$\c$-arrow $f:A\rightarrow B$ the following diagram commutes
\[\xymatrix{
A\ar[dd]^f&&Y(B)\ar[rr]^{N_B}\ar[dd]_{Y(f)}&&X(B)\ar[dd]^{X(f)}\\
&&&&\\
B&&Y(A)\ar[rr]_{N_A}&&X(A)\\
}\]
i.e.
\begin{equation*}
N_A \circ Y(f)=X(f) \circ N_B
\end{equation*}
\end{Definition}
Here $N_A:Y(A)\rightarrow X(A)$ are the components on $N$, while $N$ is the {\it natural transformation}.
From this diagram, it is clear that the two arrows $N_A$ and $N_B$ turn the $Y$-picture of $f:A\rightarrow B$
into the respective $X$-picture.
If each $N_A$ ($A\in\c$) is an isomorphism, then $N$ is a {\it natural isomorphism}
\be
N:Y\xrightarrow{\simeq}X
\ee
The collection of all contravariant functor, with natural transformations as morpshims, forms a topos which is denotes as $\Sets^{\c^{op}}$ . Roughly speaking a topos is a category which has a lot of extra structure. This extra structure implies that a general topos `looks like' $\Sets$ in the sense that any mathematical operation which can be done in set theory has a general topos analogue (ex: disjoint union, cartesian product etc.).\\\\
The most important properties of a {\it topos} for our purpose are:
\begin{enumerate}
\item The existence of a terminal object, with associated global elements.
\item The existence of a subobject classifier $\Omega$ which represents the generalisation of the set $\{0,1\}\simeq\{\text{ true, false}\}$ of truth-values in the category $\Sets$. As the name suggests, the subobject classifier identifies subobjecs. In the case of $\Sets$, given a set $A$ we say either $A\subseteq X$ or $A\nsubseteq X$. Thus to the proposition ``$A$ a subset of $X$" we can ascribe either the value true ($1$) or false $(0)$, respectively. Moreover, if $A\subseteq X$ one can ask which points $x\in X$ lie also in $A$. This can be expressed mathematically
using the so called characteristic function: $\chi_A:S\rightarrow\{0,1\}$, which is 
defined as follows:
\be\label{equ:character}
\chi_A(x)=\begin{cases} 0 \text{ if } x\notin A\\
1 \text{ if }x\in A  
\end{cases}
\ee
where here $\{0,1\}=\{\text{false, true}\}$. Thus, in sets, we only have true or false as truth values, i.e $\Omega=\{0,1\}$. This type of truth values determines the internal logic of set to be Boolean, i.e. $S\vee\neg S=1$.\\
However, in a general topos $\Omega\neq\{0,1\}$ but it will be a more generalised object (not necessarily a set) leading to a multivalued logic. In this setting we obtain a well defined mathematical notion of what it means for an object to nearly be a subobject of a given object and how far it is from being a subobject. Thus the role of a sbobject classifier $\Omega$ in a topos is to define how subobjects fit in a given object.\\
The elements of this object $\Omega$, similarly as was the case in $\Sets$, represent the truth values. The collection of all such truth values undergoes a Heyting algebra (see below). 
\item The existence of an internal logic derived from the collection of all sub-objects of any object in the topos. Such a logic is called a {\it Heyting algebra}, which is distributive and for which the law of excluded middle does not hold, i.e. $S\vee \neg S\leq 1$. 
This represents a generalisation of the Boolean algebra in $\Sets$ for which $S\vee \neg S= 1$.\\
An example of Heyting algebra is given by the collection of all open sets in a topological space
\end{enumerate}

A terminal object in the Topos of presheaves $\Sets^{\c^{op}}$ is defined as follows:
\begin{Definition}
A \textbf{terminal object} in $\Sets^{\c^{op}}$ is the constant functor 
$1:\c\rightarrow \Sets$
that maps every $\c$-object to the one element  set $\{*\}$ and every 
$\c$-arrow to 
the identity arrow on $\{*\}$.
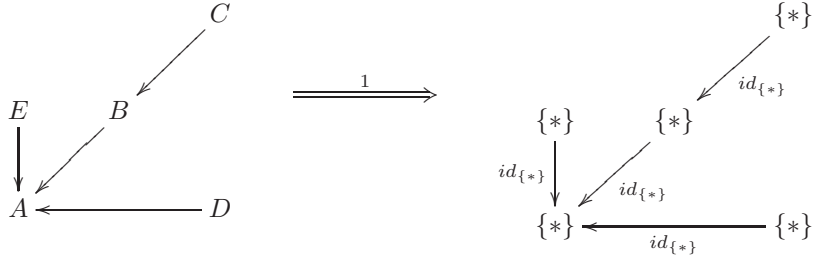
\[\xymatrix{
&&C\ar[dl]\\
E\ar[d]&B\ar[dl]&\\
A&&D\ar[ll]\\
}\hspace{.2in}
\xymatrix{
&&\\
\ar@{=>}[rr]^{1}&&}\hspace{.2in}
\xymatrix{
&&\{*\}\ar[dl]^{id_{\{*\}}}\\
\{*\}\ar[d]_{id_{\{*\}}}&\{*\}\ar[dl]^{id_{\{*\}}}&\\
\{*\}&&\{*\}\ar[ll]^{id_{\{*\}}}\\
}\]

\end{Definition}

The notion of terminal object allows us to define the notion of global element (or global section).
\begin{Definition} \label{def:glo}
A \textbf{global section} or \textbf{global element} of a presheaf $X$ in $\Sets^{\c^{op}}$ is an arrow 
$k:1\rightarrow X$ from the terminal object 1 to the presheaf $X$.
\end{Definition} 
What $k$ does is to assign to 
each object A in $\c$ an element $k_A\in X(A)$ in the corresponding set of the presheaf X.  
The assignment is such that, given an arrow $f:B\rightarrow A$ the following relation holds
\begin{equation}
X(f)(k_A)=k_B  \,.   \label{eq:k}
\end{equation}
The above relation \eqref{eq:k} uncovers that the fact that 
the elements of $X(A)$ assigned by the global section $k$ are mapped 
into each other by the morphisms in $X$.
Presheaves with a local or partial section can exist even if they do not have a global section.

\section{State space and sub-objects }
\label{app:stat}
We first discuss the topos analogue of state space and its sub-objects
in this section.

Let $\mn$ be a von Neumann algebra 
and $\cvn$ be the category of abelian von Neumann sub-algebras
of $\mn$. We consider also $\Sets$ the category of sets. 
The main topos considered here is $\Setn$, set of all 
contravariant functors from $\cvn \to \Sets$.

\subsection{State space}
To start, let us introduce basic ingredients in order to define the topos state space.

\begin{Definition}
\label{appdefsetn}
The topos $\Setn$ is the collection of all contravariant functors (called presheaves) from $\cV(\mn)$ to $\Sets$. A general element $T\in \Setn$ is defined by:
\begin{itemize}
\item[(i)] Objects: Given an object $V$ in $\cV(\mn)$, the associated set is $T(V)=T_V\in \Sets$.
\item[(ii)] Morphisms: Given a morphism $i_{V'V}:V'\rightarrow V$ ($V'\subseteq V$) in $\cV(\mn)$, the associated morphism is $T(i_{V'V}):T(V)\rightarrow T(V')$. 
\end{itemize}
\end{Definition}

The topos state space is the object in $\Setn$ defined as follows:
\begin{Definition}
\label{appdefus}
The spectral presheaf, $\us$, is the contravariant functor from the
category $\cV(\mn)$ to $\Sets$  defined
by:
\begin{itemize}
\item[(i)] Objects: Given an object $V$ in $\cvn$, the associated set $\us(V)=\us_V$ is defined to be the Gel'fand spectrum of the (unital) commutative von Neumann sub-algebra $V,$ i.e. the set of all multiplicative linear functionals $\lambda:V\rightarrow \Cl$, such that $\lambda(\hat{1})=1$, where $\hat{1}$ is the unit operator on $\mn$,
and that $\lambda(ab)=\lambda(a)\lambda(b)$, for  $a,b\in V$. 
We simply write: $\us_V=\{\lambda: V\to \Cl|\; \lambda(\hat{1})=1\}$.

\item[(ii)] Morphisms: Given a morphism $i_{V'V}:V'\rightarrow V$ ($V'\subseteq V$) in $\cvn$, the associated morphism $\us(i_{V'V}):\us_V\rightarrow
\us_{V'}$ is defined for all $\lambda\in\us_V$ to be the
restriction of the functional $\lambda:V\rightarrow\Cl$ to the
sub-algebra $V'\subseteq V$, i.e.
$\us(i_{V'V})(\lambda):=\lambda|_{V'}$.
\end{itemize}
\end{Definition}

Note that since our base category $\cvn$ is a partially ordered set, then we have an equivalence of topoi: $\Setn\simeq \Sh(\cvn)$ where $\Sh(\cvn)$ is the topos of sheaves over $\cvn$. A sheaf is nothing but a presheaf  carrying topological information and can be defined as an etal\'e bundle. In this paper, we will denote presheaves as $\ps{X}$ while the corresponding etal\'e bundles as $X$. The interested reader 
is referred to \cite{andreas5}. 

\begin{Definition}
\label{appdefsub}

A sub-object $\ps{S}$ of the spectral presheaf $\us$ is a contravariant functor  $\ps{S} :\cvn\rightarrow \Sets$ such that:
\begin{itemize}
\item[(i)] $\ps{S}_V$ is a subset of $\us_V$ for all $V\in\cvn$.
\item[(ii)] Given a map $i_{V'V}:V'\to V$, then 
$\ps{S}(i_{V'V} ): \ps{S}_V \rightarrow \ps{S}_{V'}$ is simply the restriction of the map $\us(i_{V'V} )$ to the subset $\ps{S}_V\subseteq \us_V$, thus it is given by $\lambda\mapsto\lambda|_{V'}$.
\end{itemize}
When each set $\ps{S}_V$ is clopen (closed and open) in $\us_V$, where the latter is equipped
with the usual compact and Hausdorff spectral topology, then $\ps{S}$ is called a clopen sub-object of $\us$.
The collection of all clopen sub-objects of $\us$ is denoted by $\Sub(\us)$.
\end{Definition}

\subsection{Projection operators and sub-objects of the state space}
\label{subsect:proj}

We now describe clopen sub-objects of $\us$ which are the topos analogues of propositions.

As a first step, we have to introduce the concept of
{\it daseinisation}.  Roughly speaking, the daseinisation procedure is
a kind of approximation scheme for operators so as to make them fit into any given context $V$.
Explicitly, consider the simple case in which we would like
to analyze the projection operator $\hat{P}$, which corresponds via the
spectral theorem to the proposition ``$A\in\Delta$''.\footnote{It should be noted that different propositions may correspond to the same projection operator, i.e. the mapping from propositions to projection operators is many to one. Thus, to account for this, one is really associating an equivalence class of propositions to each projection operator. The reason why von Neumann algebras were chosen instead of general C$^*$ algebras is precisely because all projections representing propositions are contained in the former, but not necessarily in the latter.} In particular, let us take a context $V$ such that $\hat{P}\notin
\mP(V)$ (the lattice of projection operators in $V$). We need to define a
projection operator which does belong to $V$ and which is related,
in some way, to our original projection operator $\hat{P}$. This
can be achieved by approximating $\hat{P}$
from above in $V$, with the {\it smallest} projection operator in $V$,
{\it greater than} or equal to $\hat{P}$. More precisely, the
{\it outer dasenisation},
 $\dase(\hat P)$, of $\hat P$ is defined at each context $V$ by
 \begin{equation}
\dase(\hat{P})_V:=\bigwedge\{\hat{R}\in
\mP(V)|\hat{R}\geq\hat{P}\}\,.
\end{equation}
Since projection operators represent propositions, $\dase(\hat{P})_V$ is a coarse graining of the proposition $``A\in \Delta"$.

This process of outer dasenisation takes place for all contexts
and, hence, gives for each projection operator $\hat{P}$, a
collection of daseinised projection operators, one for each
context V, i.e.,
\begin{align}
\hat{P}\mapsto\{\dase(\hat{P})_V|\,V\in\cvn \}\,.
\end{align}
Due to the Gel'fand transform, to each operator $\hat{P}\in
\mP(V)$ is associated the map $\bar{P}:\us_V\rightarrow\Cl$,
which takes values in $\{0,1\}\subset\Rl\subset\Cl$, since
$\hat{P}$ is a projection operator. Thus, $\bar{P}$ is a
characteristic function of the subset  
$S_{\hat{P}_V}\subseteq\us(V)$ defined by
\begin{equation}
S_{\hat{P}_V}:=\{\lambda\in  \us_V|\,\bar{P}(\lambda):=\lambda(\hat{P})=1\}\,.
\end{equation}
Since $\bar{P}$ is continuous with respect to the spectral
topology on $\ps\Sigma(V)$, then $\bar{P}^{-1}(1)=S_{\hat{P}_V}$ is a clopen subset of
$\us(V)$, since both $\{0\}$ and $\{1\}$ are clopen
subsets of the Hausdorff space $\Cl$.

Through the Gel'fand transform, it is then possible to define a
bijective map between projection operators, $\dase(\hat{P})_V\in
\mP(V)$, and clopen subsets of $\us_V$ where,
for each context V,
\begin{equation}
S_{\dase(\hat{P})_V}:=\{\lambda\in\us_V|\,\lambda
(\dase(\hat{P})_V)=1\}\,.
\end{equation}
This correspondence between projection operators and clopen
sub-objects of the spectral presheaf $\us$, or simply $\Sub(\us)$,  implies the
existence of a lattice homeomorphism for each $V$
\begin{equation}
\label{equ:smapfin}
\mathfrak{S}:\mP(V)\rightarrow \Sub(\us)_V\,,
\end{equation}
such that
\begin{equation}
\dase(\hat{P})_V\mapsto
\mathfrak{S}(\dase(\hat{P})_V):=S_{\dase(\hat{P})_V}\,.
\end{equation}
The collection of subsets
$\{S_{\dase(\hat{P})_V}\}$, $V\in\cvn$, induces a sub-object
of $\us$. 

\noindent
\begin{Theorem}
For each projection operator $\hat{P}\in \mP(\mn)$, the collection
 \be
\ps{\delta(\hat{P})}:=\{S_{\dase(\hat{P})_V}|\,V\in\cvn\}\,
\ee

forms a (clopen) sub-object of the spectral presheaf.
\end{Theorem}
\begin{proof}
We  know that, for each $V\in\cvn$, $S_{\dase(\hat{P})_V}\subseteq \us_V$. Therefore, what we need to show is that these clopen subsets get mapped one to another by the presheaf morphism. Indeed, consider an element $\lambda\in S_{\dase(\hat{P})_V}$ and given any $V'\subseteq V$, then by the definition of daseinisation, we get $\dase(\dase(\hat{P})_{V})_{V'}=\bigwedge\{\hat{\alpha}\in \mP(V')|\hat{\alpha}\geq  \dase(\hat{P})_V\}\geq \dase(\hat{P})_V$. Therefore, if $\dase(\hat{P})_{V'}-\dase(\hat{P})_V=\hat{\beta}$, then 
$\lambda(\dase(\hat{P})_{V'})=\lambda(\dase(\hat{P})_V)+\lambda(\hat{\beta})=1$ since $\lambda(\dase(\hat{P})_V)=1$ and $\lambda(\hat{\beta})\in\{0,1\}$. Moreover,
\be
\{\lambda|_{V'}|\, \lambda\in S_{\dase(\hat{P})_V}\} \subseteq\{ \lambda\in S_{\dase(\hat{P})_{V'}} \}\,.
\ee
However, $\lambda|_{V'}$ is precisely $\us(i_{V'V})\lambda$
and so one infers 
\be
\{\lambda|_{V'}|\, \lambda\in S_{\dase(\hat{P})_V}\}=\us(i_{V'V})S_{\dase(\hat{P})_V}\,.
\ee
It follows that $\ps{\delta(\hat{P})}$ is a sub-object of $\us$.
\end{proof}

\section{Phase space measure}
\label{app:mes}

In this section, we will investigate the close relationship between states $\rho$ on a von Neumann algebra and measures on the state space $\us$. In particular, we will show how, given a state $\rho$ it is possible to uniquely define a measure $\mu^{\rho}$ and, on the other hand, given an abstract characterization of a measure it is possible to uniquely determine a state (\cite{doring}).

\subsection{Deriving a measure from a state}
\label{subapp:dermeas}

To define the notion of measure on the state space $\us$,  the notion of probability valued presheaf is required.
\begin{Definition}
\label{appdef01}
The presheaf $\ps{[0,1]}:\cvn\rightarrow \Sets$ is defined on:
\begin{enumerate}
\item[(i)] Objects: For each context $V\in \cvn $, the set 
$\ps{[0,1]}_V:=\{f:\,\dV\rightarrow [0,1]|\,f \text{ is order reversing} \}$ is a set of order reversing (OR) functions over $\dV$, the downwards set of all abelian von Neumann sub-algebras of $V$.
\item[(ii)] Morphisms: Given $i_{V'V}:V'\to V,$ $V' \subset V,$ then the corresponding presheaf map is 
\ba
\ps{[0,1]}(i_{V'V}):\ps{[0,1]}_V&\rightarrow& \ps{[0,1]}_{V'}\crcr
f&\mapsto&f|_{\dV'}
\ea
\end{enumerate}
where $f|_{\dV'}$ is the restriction of $f$ to $\dV'$. 
\end{Definition}

We stress the fact that $\ps{[0,1]}$ is usually denoted
by $\ps{[0,1]}^{\succeq}$ in topos literature to distinguish it from $\ps{[0,1]}^{\preceq}$, which is the presheaf assigning to each $V$ the set of order preserving functions from $\dV$ to $\Cl$. 

\begin{Definition}
\label{appdefgam}
The collection $\ps{\Gamma[0,1]}$ of global elements (global sections) is a collection of natural transformations from the terminal object\footnote{The terminal object $\ps{1}$ is the presheaf which assigns, to each context $V$, a singleton $\{*\}$ and, to each morphism, simply the identity map.} $\ps{1}\in\Setn$ to the presheaf $\ps{[0,1]}$. Thus an element $\gamma\in \ps{\Gamma[0,1]}$ is defined, for each $V,$ as
\ba
\gamma_V:\ps{1}_V&\rightarrow& \ps{[0,1]}_V\crcr
\{*\}&\mapsto&\gamma_V(\{*\})
\ea
where $\gamma_V(\{*\}):\,\dV\rightarrow \Cl$. 
Global sections mean that, for $V'\subseteq V$, 
\be
\gamma_{V'}\{*\}=\ps{[0,1]}(i_{V'V})\gamma_V\{*\}\,.
\ee
\end{Definition}
We are now interested in constructing a measure $\mu$ on the state space $\us$. As such, it should somehow define a size or weight for each sub-object of $\us$. However, we will not consider all sub-objects of $\us$ but only a collection of them, which we define as {\it measurable}. This collection of measurable sub-objects of $\us$ will be the collection of all clopen sub-objects  $\Sub(\us)$. The reasons behind this 
choice come from classical physics ideas. In the latter context,  a proposition is defined as a measurable subset of the state space. Hence, in the topos formulation of quantum theory, we ascribe the status of measurable to all clopen sub-objects of $\us$, {\it viz.} all propositions. It could be the case that one can consider a larger collection of sub-objects but, surely, $\Sub(\us)$ is the minimal such collection. 
In any case, we will define a probability measure on $\Sub(\us)$ such that $\us$ will have measure 1 and $\ps{0}$ (the {\it empty} presheaf) will have measure 0. 
As we will see, the measure defined on $\us$ will be such that there exists a bijective correspondence between such a measure and states of the quantum system, i.e. $\mu\mapsto\varrho_{\mu}$ is a bijection\footnote{At this point it should be noted that the correspondence between measures and states is present also in the context of classical physics. In fact, in that situation,  a pure state (i.e. a point) is identified with the Dirac measure, while a general state is simply a probability measure on the state space. Such a probability measure assigns a number in the interval $[0,1]$ called the weight 
and, thereby, gives sense of the measurable set.} (for $\varrho$ a density matrix).
The remaining question is: How is this measure defined?  A suitable measure on $\us$ is defined as follows:

Given a state $\rho$, we have:
\ba
\mu_{\rho}: \Sub(\us)&\rightarrow& \Gamma\ps{[0,1]}\crcr
\ps{S}=(S_V)_{V\in\cvn}&\mapsto&\mu_{\rho}(\ps{S}):=(\rho( P_{\ps{S}_V}))_{V\in\cvn}
\label{ali:measure}
\ea
Recall that $P_{\ps{S}_V}=\mathfrak{S}^{-1}\ps{S}_V$, where 
$\mathfrak{S}:\mathcal{P}(V)\rightarrow \Sub(\us)_V$
is defined as in \eqref{equ:smapfin} (see also \eqref{equ:smap}).

Thus,
 the measure $\mu_{\rho}$ defined in \eqref{ali:measure} takes a clopen sub-object of $\us$ and defines an OR  function 
\ba
\mu_{\rho}(\ps{S}): \cvn  &\rightarrow & \ps{[0,1]}  \crcr
\;[\mu_{\rho} (\ps{S})](V): \ps{1}_V& \rightarrow& \ps{[0,1]}_V \crcr
\{*\} &\mapsto& [\mu_{\rho}(\ps{S})] (V)(\{*\}):=
\;[\mu_{\rho}(\ps{S})](V)
\ea
where $[\mu_{\rho}(\ps{S})](V):\dV\rightarrow [0,1]$.
It is worth mentioning some properties of $\mu_{\rho}$:
\begin{enumerate}
\item[(i)] For each $V\in \cvn$, then $[\mu_{\rho}(\ps{0})](V)=\rho(\hat{0})=0$ therefore, globally,
\be
\mu_{\rho}(\ps{0})=(0)_{V\in \cvn}\,.
\ee
\item[(ii)]  For each $V\in \cV(\mh)$, 
$[\mu_{\rho}(\us)](V)=\rho(\hat{1})=1$, therefore, globally,
\be
\mu_{\rho}(\us)=(1)_{V\in \cvn}\,.
\ee
\item[(iii)] Given two disjoint clopen sub-objects $\ps{S}$ and $\ps{T}$ of $\us$, we then have for each context $V\in \cvn$ that
\ba\label{ali:additive}
[\mu_{\rho}(\ps{S}\vee\ps{T})](V)&=&\mu_{\rho}((\ps{S}\vee\ps{T})_V)\crcr
&=&\mu_{\rho}(\ps{S}_V\cup\ps{T}_V)\crcr
&=&\rho(\hat{P}_{\ps{S}_V}\vee\hat{P}_{\ps{T}_V})\crcr
&=&\rho((\hat{P}_{\ps{S}_V}+\hat{P}_{\ps{T}_V})\crcr
&=&\rho(\hat{P}_{\ps{S}_V})+\rho(\hat{P}_{\ps{T}_V})\crcr
&=&[\mu_{\rho}(\ps{S})](V)+[\mu_{\rho}(\ps{T})](V)\,.
\ea
It follows, globally, that
\be
\mu_{\rho}(\ps{S}\vee\ps{T})=\mu_{\rho}(\ps{S})+\mu_{\rho}(\ps{T})\,.
\ee
This is the property of finite additivity.
\item[(iv)] As generalization of the above property, one has: given two arbitrary clopen sub-objects $\ps{S}$ and $\ps{T}$ of $\us$, for all $V\in \cvn$, the following is valid:
\ba
[\mu_{\rho}(\ps{S}\vee\ps{T})](V)+[\mu_{\rho}(\ps{S}\wedge\ps{T})](V)&=&\rho(\hat{P}_{\ps{S}_V}\vee\hat{P}_{\ps{T}_V})
+\rho(\hat{P}_{\ps{S}_V}\wedge\hat{P}_{\ps{T}_V})\crcr
&=&\rho(\hat{P}_{\ps{S}_V}\vee\hat{P}_{\ps{T}_V}
+\hat{P}_{\ps{S}_V}\wedge\hat{P}_{\ps{T}_V})\crcr
&=&\rho(\hat{P}_{\ps{S}_V}+\hat{P}_{\ps{T}_V})\crcr
&=&\rho(\hat{P}_{\ps{S}_V})+
\rho(\hat{P}_{\ps{T}_V})\crcr
&=&[\mu_{\rho}(\ps{S})](V)+[\mu_{\rho}(\ps{T})](V)\,.
\ea
Globally, this yields
\be
\mu_{\rho}(\ps{S}\vee\ps{T})+\mu_{\rho}(\ps{S}\wedge\ps{T})=\mu_{\rho}(\ps{S})+\mu_{\rho}(\ps{T})\,.
\ee
\item[(v)] 
Because the collection of clopen sub-objects of $\us$ forms  a bi-Heyting algebra and so a Boolean algebra (this has been discussed in \cite{doringx}), one has
\be
\ps{S}\vee\neg\ps{S}<\us\,,
\ee
so that 
\be
\mu_{\rho}(\ps{S}\vee\neg\ps{S})<1_{V\in\cvn}\,.
\ee
\item[(vi)] The closest property from $\sigma$-additivity is the following: given a countable infinite family $(\ps{S}_i)_{i\in I}$ for open sub-objects, such that for each context $V\in \cvn$ the clopen subsets $(\ps{S}_i)_V\subseteq \us_V$ (for all $i\in I$) are pairwise disjoint, then we have
\be
\Big[\mu_{\rho}\Big(\bigvee_{i\in I}\ps{S}_i)\Big)\Big](V)=\rho(\bigvee_{i\in I}\hat{P}_{\ps{S}_i})=\rho(\sum_{i\in I}\hat{P}_{\ps{S}_i})
=\sum_{i\in I}\rho(\hat{P}_{\ps{S}_i})=
\sum_{i\in I}[\mu_{\rho}(\ps{S}_i)](V)\,.
\ee
Thus,  $\mu$ is locally $\sigma$-additive.

\end{enumerate}

\subsection{Deriving a state from a measure}
\label{subsect:statfromes}

So far, we have defined a particular measure given a quantum state $\rho$. We are now interested in establishing the reciprocal. Thus, we first give an abstract characterization $\mu$ of a measure with no reference to a state and, then, show how such a measure can uniquely be inferred from a state $\rho_\mu$.
We end up with a bijective correspondence between the space of states $\rho$ and the space of measures $\mu$ on $\us$.

\begin{Definition}\label{def:measure}
A measure $\mu$ on the state space $\us$ is a map
\ba
\mu: \Sub(\us)&\rightarrow&\Gamma\ps{[0,1]}\crcr
\ps{S}=(\ps{S}_V)_{V\in\cvn}&\mapsto&(\mu(\ps{S}_V))_{V\in\cV(\mh)}
\ea
such that the following conditions hold:
\begin{enumerate}
\item[(i)] $\mu(\us)=(1)_{V\in \cvn}$;
\item[(ii)] For all $\ps{S}$ and $\ps{T}$ in $\Sub(\us)$, 
then $\mu(\ps{S}\vee\ps{T})+\mu(\ps{S}\wedge\ps{T})
=\mu(\ps{S})+\mu(\ps{T}).$

\end{enumerate}
\end{Definition}

\begin{Theorem}\label{theo:measure}
Given a measure $\mu$ on $\us$,  there exists a unique state $\rho_\mu$ associated to that measure.
\end{Theorem}
\begin{proof}
We define 
\ba
m:\mP(\mn)&\rightarrow& [0,1]\crcr
P&\mapsto&m(P)
\ea
such that i) $m(\hat{1})=1$; ii) if $\hat{P}\hat{Q}=0$ (are orthogonal), then $m(\hat{P}\vee\hat{Q})=m(\hat{P}+\hat{Q})=m(\hat{P})+m(\hat{Q})$. Such a map is called in the literature a {\it finitely additive probability measure on the projections of $V$}. We now want to define a unique such finite additive measure, given the probability measure $\mu$. To this end, we define
\be\label{equ:defmu}
m(\hat{P}):=[\mu(\ps{S}_{\hat{P}})](V)= \mu(\ps{S}_{\hat{P}})_V\,,
\ee
such that $m(\hat{1}):=1$. In the above formula, $\ps{S}$ is some clopen sub-object $\ps{S}\subseteq \us$, such that, for $V$, we have $\mathfrak{S}^{-1}(\ps{S}_V)=\hat{P}_{\ps{S}_V}$. 

This statement is well defined if it proves to be independent of which sub-object of $\us$ is chosen to represent the projection operator. In fact, it could happen that for contexts $V$ and $V'$, $\ps{S}_V$ and $\ps{T}_{V'}$ both correspond to the same projection operator $\hat{P}=\mathfrak{S}^{-1}(\ps{S}_V)=\mathfrak{S}^{-1}(\ps{T}_{V'})$. 
For this situation, one must show that $\mu(\ps{S}_{\hat{P}_V})=\mu(\ps{T}_{\hat{P}_{V'}})$. 
Let us first take the case for which there exist two sub-objects of $\us$, such that they correspond to the same projector at the same context.
 It can be obtained that
\ba
[\mu(\ps{S})](V)+[\mu(\ps{T})](V)&\stackrel{\ref{def:measure}}{=}&[\mu(\ps{S}\vee\ps{T})](V)+[\mu(\ps{S}\wedge\ps{T})](V)\crcr
&=&\mu[(\ps{S}\vee\ps{T})_V]+\mu[(\ps{S}\wedge\ps{T})_V]\crcr
&=&\mu(\ps{S}_V\cup\ps{T}_V)+\mu(\ps{S}_V\cap\ps{T}_V)\crcr
&=&\mu(\ps{S}_V)+\mu(\ps{S}_V)\crcr
&=&[\mu(\ps{S})](V)+[\mu(\ps{S})](V)\,.
\ea
Therefore, as expected,
\be\label{equ:samecontext}
[\mu(\ps{S})](V)=[\mu(\ps{T})](V)\,.
\ee
Next, let us consider the case in which $\ps{S}_V=\ps{T}_{V'}$ and 
that these sub-objects both correspond to the projection operator $\hat{P}$. This means that $\hat{P}\in V$ and $\hat{P}\in V'$, 
so that $\hat{P}\in V\cap V'$. For instance,  $\ps{\dase(\hat{P})}_V=\ps{S}_V=\ps{T}_V$.  Note that although $\dase(\hat{P})_{V'\cap V}=\hat{P}$, however, it is not true that $\mathfrak{S}^{-1}(\ps{S}_{V'\cap V})=\hat{P}$ or $\mathfrak{S}^{-1}(\ps{T}_{V'\cap V})$. Given such a situation, we obtain
\ba
[\mu(\ps{S})](V)&\stackrel{\ref{equ:samecontext}}{=}&
[\mu(\dase(\hat{P}))](V)\crcr
&=&[\mu(\dase(\hat{P}))](V'\cap V)\crcr
&=&[\mu(\dase(\hat{P}))](V')\crcr
&\stackrel{\ref{equ:samecontext}}{=}&[\mu(\ps{T})](V')\,.
\ea
Hence, Definition \ref{equ:defmu} makes sense. It remains to prove that $m$ satisfies ordinary properties of measures on $\mn$:
(i) $m$ is actually a finitely additive probability measure on the projections in $\mP(\mn)$. To this end, consider two orthogonal projection operators $\hat{P}$ and $\hat{Q}$, such that $\mathfrak{S}^{-1}(\ps{T}_V)=\hat{Q}$ and $\mathfrak{S}^{-1}(\ps{S}_V)=\hat{P}$. Then, $(\ps{S}\vee\ps{T})_V$ corresponds to the projection operator $\hat{P}\vee\hat{Q}$, such that
\ba
m(\hat{P}\vee\hat{Q})&:=&[\mu(\ps{S}\vee\ps{T})](V)\crcr
&=&[\mu(\ps{S})](V)+[\mu(\ps{T})](V)+[\mu(\ps{S}\wedge\ps{T})](V)\crcr
&=&[\mu(\ps{S})](V)+[\mu(\ps{T})](V)\crcr
&=:&m(\hat{P})+m(\hat{Q})\,;
\ea
(ii) the fact that $m(\hat{0})=0$ is easily inferred. Putting these results together, one can prove that, given a measure $\mu$ on $\us$, a unique   finitely additive probability measure on the projections in $\mn$ is defined. Through Gleason's theorem \cite{Gleason},\footnote{Such a theorem shows that a quantum state is uniquely determined by the values it takes on projections. Since clopen sub-objects have components which correspond to projections, Gleason's theorem applies.} it is possible to show that such a probability measure corresponds to a state $\rho_{\mu}$. Thus, the following chain $\mu\mapsto m\mapsto\rho_{\mu}$ induces 
one-to-one mappings.
\end{proof}

The proof that there exists a one-to-one correspondence between states and measures, as defined in \ref{def:measure}, has been established.

\section{Truth object and the topos $\Sh(\cvn\times (0,1)_L)$}
\label{app:truth}

We now give the definition of a truth object which is the topos analogue of a state. In particular, for any state $\rho$ and any number $r\in [0,1]$, the associated topos state is represented by the following:
\begin{Definition}
The presheaf $\ps{\mathbb{T}}^{\rho, r}\in \Setn$ is defined on
\begin{enumerate}
\item [(i)] Objects: for each context $V\in \cV(\mn)$, we have
\be\label{equ:rho}
\ps{\mathbb{T}}^{\rho, r}_V:=\{\ps{S}\in \Sub(\us_{|\dV})|\,\forall V'\subseteq V,\,  \rho(\hat{P}_{S_{V'}})\geq r\}\,.
\ee
\item [(ii)] Morphisms: given $i_{V'V}:V'\rightarrow V$, $V'\subseteq V$, the corresponding morphism is $\ps{\mathbb{T}}^{\rho, r}(i_{V'V}):\ps{\mathbb{T}}^{\rho, r}_V\rightarrow \ps{\mathbb{T}}^{\rho, r}_{V'}$; $\ps{S}\mapsto\ps{S}|_{\dV'}$.
\end{enumerate}
\end{Definition}
It can be shown that, if we have two distinct numbers $r_1\le r_2\leq 1$, then 
\be
\ps{\mathbb{T}}^{\rho, r_2}\subseteq\ps{\mathbb{T}}^{\rho, r_1}\,.
\ee
This means that the collection of sub-objects (propositions) of the state space, which are true with probability at least $r_1$, includes the collection of sub-objects (propositions) of the state space which are true with a greater probability (at least $r_2$). 

As can be seen from the above definition, we end up with a one parameter family of truth objects, namely, one for each $r$. 
We now would like to combine such a family together so as to be itself an
object in a topos. To this end one needs to extend the topos from $\Sh(\cvn)$ to $\Sh(\,\cvn \times (0,1)_L)$. 

Consider  the topological space $(0,1)$ whose open sets are the intervals $(0,r)$ for $0\leq r\leq 1$. This topological space is denoted as $(0,1)_L$ and the collection of open sets as $\mathcal{O}((0,1)_L)$. This is a category under inclusion. One can then define a bijection
\ba
\iota:[0,1]&\rightarrow& \mathcal{O}((0,1)_L)\\
r&\mapsto&(0,r)
\ea
In this setting, the stages/contexts are now pairs $(V, r)$. Such a category $\cvn\times (0,1)_L$ can be given the structure of a poset as follows:
\be
(V', r')\leq (V, r) \qquad \text{ iff } 
\qquad V'\leq V \;\text{ and }\; r'\leq r\,.
\label{vrvr}
\ee
The new truth object now becomes, for each context 
$(V, r)$,
\be
\tob^{\rho}_{\Sh(\,\cvn \times (0,1)_L)}((V, r)):= 
\tob^{\rho, r}(V)=\{\ps{S}\in \Sub(\us_{|\dV})|\;\forall V'\subseteq V\;,\; \rho(\hat{P}_{\ps{S}_{V'}})\geq r\}\,.
\ee
As shown in \cite{probabilities} each object in the topos $\Setn$ can be consistently translated to an object in the topos $\Sh(\,\cvn\times (0,1)_L)$, without any information being lost.

\end{document}